\documentclass[12pt,dvipsnames]{article}
\usepackage{setspace}

\usepackage[utf8]{inputenc}
\usepackage{jheppub}
\usepackage{epsfig}
\usepackage{bbm}
\usepackage{bbding}
\usepackage{amsmath}
\usepackage{amsfonts}
\usepackage{amssymb}
\usepackage{mathtools}
\usepackage{dsfont}
\usepackage{bm}

\usepackage{graphicx}
\usepackage{longtable}
\usepackage{pdflscape}
\usepackage{subcaption}
\usepackage{psfrag}
\usepackage{hyperref}
 
\usepackage[capitalise]{cleveref}
\usepackage{tikz}
\usetikzlibrary{positioning,arrows.meta}
\usepackage{pgfplots}
\usepackage{pgfplotstable}
\usepackage{subfiles}
\usepackage{longtable}

\usepackage{tabularx}
\usepackage{ltablex}
\usepackage{booktabs}
\usepackage[export]{adjustbox}
\usepackage{listings}
\usepackage{multirow} 
\usepackage{cancel,slashed}
\usepackage{bbm}
\usepackage{graphicx}
\usepackage{latexsym}
\usepackage{transparent}
\usepackage{mathtools}
\usepackage{array}
\usepackage{makecell}
\usepackage{graphics,psfrag}
\usepackage{placeins}
\usepackage{nowidow}
\usepackage{listings}
\usepackage[normalem]{ulem}
\usepackage{environ}
\usepackage{tikz}
\usepackage{enumitem}   
\usetikzlibrary{positioning,arrows.meta}
\usetikzlibrary{positioning,trees,decorations.pathmorphing,decorations.markings,decorations.pathreplacing,calc,shapes,shapes.geometric,shapes.symbols,patterns,arrows}
\usetikzlibrary{fadings}
\usetikzlibrary{decorations.shapes}
\usepackage{amsmath}
\usepackage{amssymb}
\usepackage{amsthm}

\newcommand{\newshortstack}[1]
{\begingroup\renewcommand{\arraystretch}{1.1}% <- adjust to suit
\ifmmode
\begin{array}{c}#1\end{array}%
\else
\begin{tabular}{c}#1\end{tabular}%
\fi
\endgroup}

\pgfplotsset{
        compat=1.9,
        compat/bar nodes=1.8,
    }

\theoremstyle{definition} 
\newtheorem{theorem}{Theorem}
\newtheorem{corollary}{Corollary}[theorem]
 
\newcommand{\be}{\begin{equation}}
	\newcommand{\ee}{\end{equation}}
\newcommand{\bea}{\begin{eqnarray}}
	\newcommand{\eea}{\end{eqnarray}}

\newcommand{\slz}{\text{SL}(2,\mathbb{Z})}
\newcommand{\pslz}{\text{PSL}(2,\mathbb{Z})}
\newcommand{\kah}{\mathcal{K}}
\newcommand{\tbar}{\bar{T}}
\newcommand{\sbar}{\bar{S}}
\newcommand{\Sb}{\bar{S}}
\newcommand{\alphap}{\alpha^\prime}
\renewcommand{\epsilon}{\varepsilon}

\newcommand{\ben}{\begin{enumerate}}
	\newcommand{\een}{\end{enumerate}}
\newcommand{\bei}{\begin{itemize}}
	\newcommand{\eei}{\end{itemize}}
	
	\makeatletter
\tikzset{
    dot diameter/.store in=\dot@diameter,
    dot diameter=3pt,
    dot spacing/.store in=\dot@spacing,
    dot spacing=10pt,
    dots/.style={
        line width=\dot@diameter,
        line cap=round,
        dash pattern=on 0pt off \dot@spacing
    }
}
\makeatother
	
\tikzset{decorate sep/.style 2 args=
{decorate,decoration={shape backgrounds,shape=circle,shape size=#1,shape sep=#2}}}
\newtheorem*{theorem*}{Theorem}

\newcommand{\mc}{\mathcal}

\newcommand{\coma}{\, , \quad}
\newcommand{\fstop}{\, .}

\def\im{{\text{Im} \,}}
\def\re{{\text{Re} \,}}

\definecolor{outerrim}{rgb}{16,133,47}

\graphicspath{{Figures/}}

\title{Heterotic de Sitter Beyond Modular Symmetry}
\author[a]{Jacob M. Leedom,}
\author[b]{Nicole Righi,}
\author[a]{$and$ Alexander Westphal}

\affiliation[a]{ Deutsches Elektronen-Synchrotron DESY\\ Notkestr. 85, 22607 Hamburg, Germany}
\affiliation[b]{ Physics Department, King's College London\\ Strand, London, WC2R 2LS, U.K.}

\emailAdd{jacob.michael.leedom@desy.de}
\emailAdd{nicole.righi@kcl.ac.uk}
\emailAdd{alexander.westphal@desy.de}

\abstract{We study the vacua of $4d$ heterotic toroidal orbifolds using effective theories consisting of an overall K\"{a}hler modulus, the dilaton, and non-perturbative corrections to both the superpotential and K\"{a}hler potential that respect modular invariance. We prove three de Sitter no-go theorems for several classes of vacua and thereby substantiate and extend previous conjectures. Additionally, we provide evidence that extrema of the scalar potential can occur inside the PSL$(2,\mathbb{Z})$ fundamental domain of the K\"{a}hler modulus, in contradiction of a separate conjecture. We also illustrate a loophole in the no-go theorems and determine criteria that allow for metastable de Sitter vacua. Finally, we identify inherently stringy non-perturbative effects in the dilaton sector that could exploit this loophole and potentially realize de Sitter vacua.}

\arxivnumber{2212.03876}

\begin{document}
\maketitle
	\pagestyle{plain}

	%----------------------------------------------------------------------%
	%  numbering equations with section number
	%----------------------------------------------------------------------%
	\makeatletter
	\@addtoreset{equation}{section}
	\makeatother
	\renewcommand{\theequation}{\thesection.\arabic{equation}}
	%----------------------------------------------------------------------%
	%  title page
	%----------------------------------------------------------------------%
	\pagestyle{empty}
	\setcounter{page}{1}
	\pagestyle{plain}
	\renewcommand{\thefootnote}{\arabic{footnote}}
	\setcounter{footnote}{0}
	%----------------------------------------------------------------------%
	%  Paper begins
	%----------------------------------------------------------------------%
	
	%\end{document}
	%\newpage
	%\tableofcontents

%===========================================
\newpage
\section{Introduction}

The accelerating expansion of the Universe~\cite{SupernovaCosmologyProject:1998vns,SupernovaSearchTeam:1998fmf} and the nature of the dark energy driving it remain some of the greatest challenges in theoretical physics. The simplest mathematical model consistent with observations is that our universe is in a de Sitter (dS) phase and dark energy corresponds to a cosmological constant in Einstein's equations. However, the scale of this constant is orders of magnitudes below naive expectations from the Standard Model of particle physics.

One might hope that natural explanations for the origin and scale of dark energy can be found in an ultraviolet complete theory of particle interactions and gravity. Our only viable candidate for such a theory is string theory, and so the quest for understanding the cosmological constant transmutes into a question of which string compactifications can yield four-dimensional dS cosmologies.

The construction of 
such cosmologies within string theory has been and remains one of the paramount tasks for string theory to connect to the low-energy physics of our universe. Existing construction schemes with still limited, but increasingly improved, control involve perturbative moduli stabilization with fluxes, branes and orientifold planes in dS vacua on negatively curved internal manifolds (such as twisted tori, product Riemann surfaces, and more general compact $6d$ hyperbolic spaces) as well as dS vacua on Calabi-Yau (CY) manifolds with fluxes and either purely non-perturbative or a mix of perturbative and non-perturbative quantum corrections stabilizing the volume moduli of the CY. For recent reviews providing a comprehensive overview as well as references to the original literature, see~\cite{Silverstein:2016ggb,Cicoli:2018kdo,Flauger:2022hie}.

While the viability of the various constructions is not completely settled, one could wonder if there exists some reason that dS vacua cannot exist at all in string theory~\cite{Danielsson:2018ztv}. Indeed, the refined dS conjecture~\cite{Garg:2018reu,Ooguri:2018wrx,Hebecker:2018vxz} as part of the Swampland program~\cite{Vafa:2005ui} posits that long periods of accelerated cosmological expansion are severely limited in quantum gravity and metastable dS minima are forbidden in asymptotic regions of the moduli space (the earlier, stronger dS conjecture~\cite{Obied:2018sgi} required modification due to explicit counter-examples~\cite{Denef:2018etk,Conlon:2018eyr,Murayama:2018lie,Choi:2018rze,Hamaguchi:2018vtv}).
 The asymptotic dS conjecture can be thought of as a generalization of the Dine-Seiberg problem \cite{Dine:1985he} in all directions of the moduli space. There is also an even weaker version of the dS conjecture, the Transplanckian Censorship Conjecture \cite{Bedroya:2019snp}, which allows for short-lived dS minima residing in the interior of moduli space.

The above statements are conjectural, but they are motivated by a number of no-go results in the literature that forbid dS vacua in specific contexts. The classical $10d$ supergravity Maldacena-Nu\~nez no-go result~\cite{Maldacena:2000mw} suggests that any dS constructions must include non-classical ingredients as D-branes, O-planes, or quantum corrections. In the context of heterotic string theories, this no-go has been pushed farther to include stringy effects. The analysis in~\cite{Covi:2008ea} provides the general criterion which the K\"ahler manifolds of the volume moduli space of CY compactifications producing $4d$, ${\cal N}=1$ supergravity effective actions must fulfill to allow for metastable dS vacua to exist. This is a generalization of the K\"ahler geometry arguments provided in~\cite{Gomez-Reino:2006tjy} based on~\cite{Brustein:2004xn}.

Next,~\cite{Green:2011cn} considered the sub-leading $\alphap$ corrections arising from the modified Bianchi identity of the Neveu-Schwarz (NS) 2-form $B_2$ and concluded that dS vacua are not possible. This was further extended to an infinite tower of $\alphap$ contributions, where a perturbative calculation showed that AdS and dS are both ruled out~\cite{Gautason:2012tb}. Finally, the powerful worldsheet argument of~\cite{Kutasov:2015eba} ruled out any worldsheet effect giving rise to dS vacua at tree level in string perturbation theory.

Apart from worldsheet effects, string compactifications will have non-perturbative contributions in the string coupling $g_s$. Focusing again on heterotic constructions, a simple example is gaugino condensation, which has a generic strength of order $\delta\mathcal{L}\sim e^{-1/g_s^2}$. In~\cite{Quigley:2015jia}, partial no-go results were obtained ruling out AdS and dS solutions. However, this was not a comprehensive argument since threshold corrections~\cite{Kaplunovsky:1987rp,Dixon:1990pc,Antoniadis:1991fh,Antoniadis:1992rq,Antoniadis:1992sa,Kaplunovsky:1995jw} and worldsheet instantons were not included, hence some heterotic constructions evade the no-go results~\cite{Cicoli:2013rwa}. There are some results incorporating gaugino condensation, threshold corrections, and worldsheet instantons in the context of toroidal orbifold compactifications of the heterotic string~\cite{Gonzalo:2018guu,Parameswaran:2010ec}. The authors of~\cite{Gonzalo:2018guu} consider only the dilaton and overall K\"{a}hler modulus of the compactification. Using target space modular symmetry to enumerate all possible non-perturbative contributions, the authors find that AdS minima are generically present, while they argue numerically that no dS solutions can be realized. Note that~\cite{Brustein:2004xn} already contained in its Section II a limited version of the statements in~\cite{Gonzalo:2018guu}. On the other hand, ~\cite{Parameswaran:2010ec} considered all bulk moduli for certain orbifolds and numerically found only unstable dS extrema that satisfied the refined dS conjecture~\cite{Olguin-Trejo:2018zun}. 

However, gaugino condensation is not the only non-perturbative contribution in $g_s$ present in the heterotic string theory. As argued by Shenker~\cite{Shenker:1990}, all closed string theories generically have effects of strength $\delta\mathcal{L}\sim e^{-1/g_s}$. These are inherently stringy effects, in contrast to the purely quantum field theoretical nature of gaugino condensation. As the above arguments apply only to $\mathcal{O}(e^{-1/g_s^2})$ effects, examining these ``Shenker-like" contributions to heterotic string vacua is the logical extension of previous no-go results.
 
 In this work, we set our sights on this task. We will consider two-modulus models of heterotic toroidal orbifold compactifications and include Shenker-like effects as non-perturbative corrections to the dilaton K\"{a}hler potential $k(S,\bar{S})$. We will prove three no-go theorems that forbid dS vacua for different branches of solutions, and affirm several conjectures in~\cite{Gonzalo:2018guu} as corollaries. However, we find that on a separate branch of solutions, the Shenker-like effects may provide a loophole to the no-go theorems and permit heterotic dS vacua.

 The paper is organized as follows. In~\cref{sec:heterotic}, we review the relevant details of standard heterotic toroidal orbifold compactifications with a focus on the effective two-modulus model and the interplay of T-duality, threshold corrections, and non-perturbative effects. Then, we prove a no-go theorem that forbids de Sitter vacua for a class of extrema, even with Shenker-like effects included in the dilaton K\"{a}hler potential. In \cref{sec:vacua}, we describe sufficient criteria to evade this no-go result and describe in detail the behavior of extrema in the fundamental domain of the K\"{a}hler modulus. In~\cref{sec:npeffects}, we prove an additional no-go theorem and further restrict the classes of models that could contain dS vacua. We then review the scant literature on Shenker-like effects and show that they in principle satisfy the criteria to evade the no-go results we establish. 
We provide preliminary examples where dS vacua can be constructed in a bottom-up fashion. We also comment on a puzzle posed in~\cite{Silverstein:1996xp} on non-perturbative effects that are even stronger than those of Shenker at weak coupling. In \cref{sec:discussion}, we conclude and discuss multiple future directions. In the appendices, we include supplemental information for the main text as well as proving a third no-go theorem in~\cref{sec:anomaly} for orbifolds with moduli mixing in the K\"ahler potential arising from anomaly cancellation.

%=================================================
\section{Heterotic Orbifolds \& A No-Go Theorem}
\label{sec:heterotic}
%=================================================
In this section, we give a brief review of heterotic toroidal orbifolds, defining the set of four-dimensional effective theories considered throughout the paper.
These consist of an overall K\"{a}hler modulus $T$ and dilaton $S$.
Non-perturbative effects such as gaugino condensation and worldsheet instantons are captured by the superpotential $W(S,T)$ while the non-perturbative Shenker-like effects are captured by the dilaton K\"{a}hler potential $k(S,\bar{S})$. We then study the extrema of the scalar potential and prove a dS no-go theorem for a particular branch of extrema.

\subsection{Simplified Heterotic Orbifold Models}
We consider $\mathcal{N}=1$ supersymmetric toroidal orbifold models of the heterotic string, i.e. the heterotic string on $T^6/G$, where $G$ is the orbifold action defined by some discrete group. For $(2,2)$ constructions arising from standard embedding, the spectrum of these models includes the dilaton, $h^{1,1}$ K\"{a}hler moduli, $h^{2,1}$ complex structure moduli, gauge fields, and twisted and untwisted matter fields. General $(0,2)$ orbifolds contain the same untwisted moduli as the $(2,2)$ compactifications but with model-dependent twisted sector moduli.

\begin{figure}[t]
    \centering
\begin{tikzpicture}[scale=2]
\draw[line width=0.5pt,black,Triangle-] (-1.2,0) -- (0,0);
\draw[line width=0.5pt,black,-Triangle] (0,0) -- (1.2,0) node[below] {\color{black}{\small{Re$(T_i)$}}};
\draw[line width=0.5pt,black] (0,0) -- (0,1);
\draw[line width=0.5pt] (0,1) -- (0,2.75) node[left,xshift=0.15cm,yshift=0.3cm] {\color{black}{\small{Im$(T_i)$}}};
\draw[dashed,line width=0.7pt] (0.5,0.866) -- (0.5,2.8);
\draw[line width=0.7pt,OliveGreen,-Triangle] (-0.5,0.866) -- (-0.5,2.8);
\draw[line width=0.7pt,OliveGreen] (-0.5,0.866) arc[start angle=120, end angle=90,radius=1cm];
\draw[fill=OliveGreen!30,path fading=north,line width=0.01pt](0,1)--(0,2.7)--(-0.5,2.7)--(-0.5,0.866) arc[start angle=120, end angle=90,radius=1cm,line width=0.001pt];
\draw[fill=OliveGreen!30,path fading=north,line width=0.001pt](0,1)--(0,2.7)--(0.5,2.7)--(0.5,0.866) arc[start angle=60, end angle=90,radius=1cm,line width=0.001pt];
\draw[dashed,line width=0.7pt] (0.5,0.866) arc[start angle=60, end angle=90,radius=1cm];
\node at (0,1)[circle,fill=Black,inner sep=1.5pt,label={[xshift=0.12cm,yshift=-0.6cm]:$i$}]{};
\node at (-0.5,0.866)[circle,fill=Black,inner sep=1.5pt,label={[xshift=0cm,yshift=-0.63cm]:$\rho$}]{};
%\node at (-1,1)[circle,fill=blue,inner sep=1pt]{};
\node at (0,1.23)[circle,fill=Black,inner sep=1.5pt,label={[xshift=0.5cm,yshift=-0.1cm]:\small{$1.23i$}}]{};
% \draw[line width=0.5pt,teal,Triangle-]
% (0,1.8) -- (1.2,2.3) node[right,yshift=0.1cm] {\color{teal}{\small{outer rim}}};
% \draw[line width=0.5pt,teal,Triangle-]
% (0.5,1.65) -- (1.2,2.3);
\end{tikzpicture}
 \caption{The strict fundamental domain $\widetilde{\mathcal{F}}_i$ of the diagonal K\"ahler modulus $T_i$. The dots indicate the fixed points as well as the oft-quoted stabilized value $T_i\simeq1.23i$ found in heterotic moduli stabilization.
} 
 \label{fig:sl2zfun}
\end{figure}
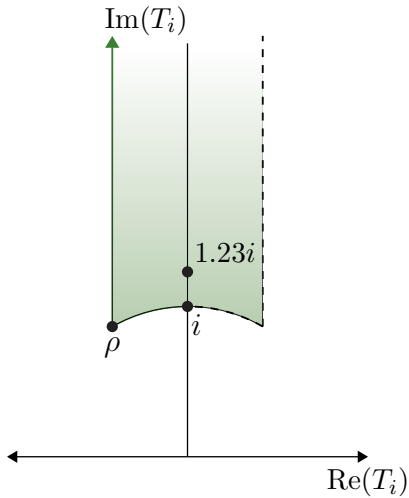

Of critical importance is that the effective field theories (EFTs) of these orbifolds respect a target-space modular symmetry arising from T-duality~\cite{Ferrara:1989bc,Font:1990nt,Cvetic:1991qm}. If we consider for the moment the $T^6$ to be a factorized product of three $T^2$, then each of the diagonal K\"{a}hler moduli $T_i$ has a $\pslz_i$ modular symmetry that acts as
\be\label{eq:slztransf}
    T_i\rightarrow \gamma_i\cdot T_i = \frac{a_iT_i+b_i}{c_iT_i+d_i}\quad \text{ with } \quad \gamma_i = \begin{pmatrix}a_i & b_i\\ c_i & d_i\end{pmatrix}\in\pslz_i\fstop
\ee
Ostensibly, the K\"{a}hler moduli are valued in the upper half planes
\be
\mathcal{H}_i = \{T_i\in\mathbb{C}\;\; | \;\;\text{Im}(T_i)>0\}
\ee
since their vacuum expectation values (vevs) control the size of the compact dimensions. However, modular symmetry restricts the set of inequivalent $T_i$ values to the strict fundamental domain $\widetilde{\mathcal{F}}_i$, which is defined as the union of the set
\be
\mathcal{F}_i = \bigg\{ T_i \in \mathcal{H}_i\;\; \rvert\;\; \lvert T_i\rvert >1 \;\;\&\;\; \lvert\text{Re}(T_i) \rvert< \frac{1}{2} \bigg\}\coma 
\ee
with boundary points that satisfy $\text{Re}(T_i)<0$. This is the region displayed in~\cref{fig:sl2zfun}. Of particular note are the fixed points $T_i=i$ and $T_i= e^{2\pi i/3}\equiv \rho$, which are fixed by cyclic subgroups of $\pslz_i$, as described in~\cref{app:modularforms}. 

The matter fields $\Phi_{\alpha}$ of the compactification also transform under modular transformations as
\be
    \Phi_{\alpha}\rightarrow  e^{\sum_in_{\alpha}^i \mathcal{J}_i(T_i)}\Phi_{\alpha}\coma
\label{eq:mattransf}
\ee
where $\mathcal{J}_i(T) = \ln(c_i T_i+d_i)$ and the numbers $n_{\alpha}^i$ are the weights of the matter field. For an arbitrary orbifold, the modular symmetry group may not be a simple product of $\pslz$ factors. For example, if one introduces discrete Wilson lines or orbifolds where the $6d$ lattice is not a simple direct sum, then the duality group will in general be a congruence subgroup of $\pslz$~\cite{Love:1996sk,Mayr:1993mq}. Irrespective of the precise form of the modular symmetry group, the effective action and the scalar potential 
\be
    V = e^{\mathcal{K}} \left(\kah^{a\bar b} F_a \bar F_{\bar b} - 3W\bar W\right)
\label{eq:sugrapot}    
\ee
must be modular functions, i.e. they should transform as modular forms of weight\footnote{See \cref{app:modularforms} for terminology and details on modular forms.} $(0,0)$: $V(\gamma\cdot T_i, \gamma\cdot \bar{T}_i,...) = V(T_i,\bar{T}_i,..)$. Here $\mathcal{K}$ and $W$ are the K\"{a}hler potential and superpotential, respectively, and the $F_a$ are the F-terms of the fields, defined below.  The above condition can also be cast as the restriction that the defining supergravity function $\mathcal{G} = \mathcal{K} + \ln\lvert W\rvert^2$ is modular invariant.

In the following, we shall neglect all of the moduli and matter fields except for an overall diagonal K\"{a}hler modulus, $T$, and the dilaton. We shall also neglect Wilson lines and therefore take the modular symmetry of the effective field theory to be $\pslz$. The K\"{a}hler potential of this two-modulus EFT is
\be
    \kah= k(S,\bar{S},T,\bar{T}) -3\ln\left(-i(T-\bar T)\right) \fstop
    \label{eq:kahpot1}
\ee
For the moment, we are using the familiar chiral multiplet formalism for the dilaton with $S$ being its chiral superfield representation consisting of a scalar (the dilaton), a pseudoscalar (from dualizing the NS 2-form $B_2$), a Weyl spinor, and an auxiliary field. The function $k(S,\bar{S},T,\bar{T})$ is the K\"{a}hler potential for the dilaton, which may have dependence on $T$, as discussed below. At tree-level, there is no $T$-dependence in the dilaton K\"{a}hler potential, hence $k(S,\bar{S},T,\bar{T}) = -\ln(S+\bar{S})$ and the $4d$ universal gauge coupling is 
\begin{equation}
    \frac{g_4^2}{2} = \bigg\langle \frac{1}{S+\bar S}\bigg\rangle\fstop
\end{equation}
Under modular transformations, the combination $T-\bar T$ transforms as a weight $(-1,-1)$ modular form since $(T-\bar T)\rightarrow \lvert  cT+d\rvert^{-2}(T-\bar T)$. Then assuming that $k(S,\bar{S},T,\bar{T})$ is inert under $\pslz$, the K\"{a}hler potential undergoes a K\"{a}hler transformation 
\be
\mathcal{K}\rightarrow \mathcal{K} + 3\mathcal{J}(T)+3\bar{\mathcal{J}}(\bar{T}),
\ee
where $\mathcal{J}(T)$ is the same as in the phase of~\cref{eq:mattransf}  but with the indices removed. The invariance condition on $\mathcal{G}$ then implies that the superpotential must transform as a weight $(-3,0)$ modular form~\cite{Ferrara:1989bc}:
\be
 W(S,\gamma\cdot T) = e^{i\delta(\gamma)} (cT+d)^{-3}W(S,T)\fstop
  \label{eq:superpottransf}
\ee
Technically, we allow the superpotential to furnish a projective representation of $\pslz$ and so it may transform as a weight $(-3,0)$ modular form up to the phase $\delta(\gamma)$, which depends only on the matrix $\gamma\in\pslz$.
%========================================================================

%========================================================================

As we are neglecting matter fields, the superpotential we are considering is non-perturbative in nature and arises from gaugino condensation of some subgroups of the $10d$ heterotic gauge sector. If the gauge group factor $G_a$ undergoes gaugino condensation, a non-perturbative superpotential is generated of the form  
\be
 W\sim e^{-f_a/b_a}\fstop
 \label{eq:superpot0}
\ee
Here $f_a$ is the gauge kinetic function for $G_a$ and the coefficient $b_a$ is related to the beta function of $G_a$ as $\mu \frac{dg_a}{d\mu} = -\frac{3}{2}b_a g_a^3$ and is given by 
\be
    b_a = \frac{1}{8\pi^2}\left(C_a - \frac{1}{3}\sum_\alpha C^\alpha_a\right)\coma
\ee
where $C_a$ and $C^\alpha_a$ are the quadratic Casimirs of the adjoint and the $\alpha$-th matter sectors, respectively. At tree level, $f_a = k_aS$, with $k_a$ the level of the Kac-Moody algebra underlying the gauge group $G_a$. In this form, it is not clear that~\cref{eq:superpot0} transforms with weight $(-3,0)$ since naively the dilaton is invariant under modular transformations. 

This situation is remedied by taking into account the 1-loop effects of threshold corrections~\cite{Kaplunovsky:1987rp,Dixon:1990pc,Antoniadis:1991fh,Antoniadis:1992rq,Antoniadis:1992sa,Kaplunovsky:1995jw} and anomaly cancellation~\cite{Derendinger:1991hq,Lust:1991yi,LopesCardoso:1991wk,LopesCardoso:1991ifk,LopesCardoso:1992yd,deCarlos:1992kox}. The physical origins of these effects are quite simple -- the fermions in the $4d$ theory undergo modular transformations and generally lead to non-zero anomalies that must be canceled by the $4d$ Green-Schwarz mechanism. Furthermore, if the orbifold has $\mathcal{N}=2$ subsectors, integrating out the heavy string states will lead to moduli-dependent corrections to the gauge kinetic function. To incorporate these contributions to the effective model, we include a correction into the gauge kinetic function 
\be
    f_a = k_a S +\left(b^\prime_a - \frac{1}{3} k_a \delta_{GS}\right) \ln\eta^6(T)+\cdots\coma
    %f_a = k_a S + \frac{1}{4\pi^2}\bigg(\frac{1}{2}b^\prime_a - k_a \delta_{GS}\bigg) \ln\eta^6(T)+\cdots\coma
\label{eq:gaugekin}
\ee
and to the dilaton K\"{a}hler potential 
\be
    -\ln (S+\bar{S})\rightarrow  -\ln\bigg(S+\bar{S} + \delta_{GS}\ln\left(-i(T-\bar T)\right)\bigg)\fstop
\label{eq:Kahanom}
\ee
The coefficient $b_a^\prime$ is determined by the quadratic Casimirs and the modular weights of the matter fields charged under $G_a$:
\be
    b_a^\prime =  \frac{1}{8\pi^2}\left(C(G^a) - \sum_\alpha C^\alpha_a(1+2n_\alpha)\right)\fstop
\ee
The dilaton now transforms as
\be
    S\rightarrow S + \delta_{GS}\ln(cT+d)
\label{eq:diltran}
\ee
under $\pslz$ in order to maintain modular invariance of the action. If we neglect matter fields, as in the case of a hidden $E_8$ condensate, then inserting~\cref{eq:gaugekin} into~\cref{eq:superpot0} yields
\be
    W\sim \frac{e^{-k_aS/b_a}}{\left(\eta(T)\right)^{6-\frac{2k_a\delta_{GS}}{b_a}}}\fstop
\label{eq:superpotp5}
\ee
As described in~\cref{app:modularforms}, $\eta(T)$ transforms with weight $(1/2,0)$, and combining this fact with~\cref{eq:diltran} we see that the superpotential transforms with weight $(-3,0)$ and the scalar potential will be invariant under modular transformations. Hence, threshold corrections and anomaly cancellation are essential ingredients in the consistency of the $4d$ EFT arising from the heterotic string.

While the above makes the consistency of modular invariance in the EFT clear, it will be convenient for our purposes to adopt a convention where the dilaton is inert under modular transformations. One is free to re-define the dilaton via a holomorphic function of other moduli~\cite{Kaplunovsky:1995jw,Derendinger:1991hq} -- in particular, we can re-define the dilaton via
\be
S\rightarrow S + \delta_{GS}\ln\eta^2(T)\fstop 
\label{eq:dilredef}
\ee
This eliminates the dependence of the superpotential on $\delta_{GS}$ so that all information on anomaly cancellation in encoded in the re-defined K\"{a}hler potential.\footnote{An alternative approach is to utilize the linear multiplet formalism for the dilaton. The dilaton is still inert under modular transformations, and anomaly cancellation is then achieved by adding a term to the action while leaving the K\"{a}hler potential untouched. We return to this point in~\cref{sec:npeffects}. } 

However, the superpotential obtained from~\cref{eq:superpotp5} with~\cref{eq:dilredef} is still too primitive -- as the dots in~\cref{eq:gaugekin} suggest, there are additional contributions to the threshold corrections~\cite{Kaplunovsky:1995tm,Kiritsis:1996dn}.\footnote{These terms include the 1-loop prepotential of the $\mathcal{N}=2$ sector, which has an interesting relation to Mathieu moonshine~\cite{Wrase:2014fja}. } These are moduli-dependent but transform trivially under $\pslz$.  One can push the power of modular invariance even further to parametrize these effects. After the dilaton re-definition in~\cref{eq:dilredef}, the Dedekind etas saturate the transformation law of the superpotential, but the numerator could involve a function of $T$ that is modular invariant (up to a phase). Then a general non-perturbative superpotential satisfying the T-duality requirement has the form
\be
 W(S,T) = \frac{\Omega(S)H(T)}{\eta^6(T)}
 \label{eq:superpot1}\fstop
\ee
Here $H(T)$ is a modular function with a potentially non-trivial multiplier system. If $H(T)$ is regular in the fundamental domain, a  theorem~\cite{10.2307/1968796,Lehner:1964} states that it has the parametrization 
\be
    H(T) = \bigg(\frac{G_4(T)}{\eta^8(T)}\bigg)^n\bigg(\frac{G_6(T)}{\eta^{12}(T)}\bigg)^m \mathcal{P}(j(T))\coma
   \label{eq:Hpara}
\ee
with $G_4(T)$ and $G_6(T)$ the weight $(4,0)$ and $(6,0)$ holomorphic Eisenstein series and $j(T)$ the $j$-invariant (see \cref{app:modularforms}). $\mathcal{P}(x)$ is a polynomial, and without loss of generality we take $\mathcal{P}(0)\neq 0$ and $\mathcal{P}(1728)\neq 0$. In~\cref{eq:superpot1}, we can technically allow an arbitrary function for $\Omega(S)$, but we will mostly consider it as arising from gaugino condensation and one should think of it as having the generic form $\Omega(S)=h+e^{-S/b_a}$ for a single condensate, with $h$ an additive constant to describe $H_3$-flux, or $\Omega(S)=\sum_a e^{-S/b_a}$ for a racetrack scenario.\footnote{A more general superpotential is
\be
    W(S,T) =\frac{1}{\eta^6(T)} \sum_a\Omega_a(S)H_a(T)\coma
\label{eq:gensuperpot}
\ee 
which corresponds to a racetrack scenario with $T$-dependent coefficients.} This superpotential was first proposed in~\cite{Cvetic:1991qm}.

Note that the terms in $H(T)$ have the schematic form of $\delta\mathcal{L} \sim e^{\pm2\pi i T}$ -- the function $H(T)$ can be thought of parametrizing non-perturbative effects in the K\"{a}hler modulus. Thus the superpotential in~\cref{eq:superpot1} parametrizes non-perturbative effects in the superpotential for both moduli in the effective model.

The above discussion defines a broad class of effective heterotic toroidal orbifold models via the superpotential in~\cref{eq:superpot1} and the K\"{a}hler potential from~\cref{eq:kahpot1} with  
\be
    k(S,\Sb,T,\tbar) \supset -\ln\bigg(S+\Sb+\delta_{GS}\ln\left(-i(T-\tbar)\lvert\eta(T)\rvert^4\right)\bigg)\fstop 
\label{eq:modinvKahpot}    
\ee
It is instructive to consider two extreme cases. For the simple $\mathbb{Z}_3$ orbifold with an $E_8$ condensate, there is no $\mathcal{N}=2$ subsector and $\delta_{GS}= 3b_{E_8}/k_{E_8}$. If we consider the formalism where the dilaton transforms under $\pslz$, we see that the exponent of the Dedekind eta in~\cref{eq:superpotp5} vanishes and the modular transformation properties of the superpotential are encoded entirely in the dilaton. In the formalism where the dilaton is invariant, the superpotential has factors of Dedekind eta, but these vanish at the level of the scalar potential when the original dilaton variable is used. These features are reflections of the usual statement that the $\mathbb{Z}_3$ (and $\mathbb{Z}_7$) orbifolds lack $\mathcal{N}=2$ subsectors and therefore have no moduli-dependent threshold corrections~\cite{Dixon:1990pc}. On the opposite end of the spectrum, the $\mathbb{Z}_2\times \mathbb{Z}_2$ orbifold with standard embedding has $\delta_{GS}=0$, so the dilaton is invariant under $\pslz$ without any redefinition and the Dedekind etas are required in the superpotential and scalar potential. These two examples illustrate the diverse ways in which heterotic orbifolds conspire to maintain modular symmetry in spite of effects that naively break the duality.

We now largely restrict ourselves to a particular subclass of the models described above. In particular, we will neglect the $T$-dependent correction to the dilaton K\"{a}hler potential 
arising from anomaly cancellation and set the dilaton-dependent term in~\cref{eq:kahpot1} to
\be
k(S,\bar{S},T,\bar{T}) = k(S,\bar{S}) = -\ln(S+\bar{S}) + \delta k(S,\bar{S})\fstop
\label{eq:minikah}
\ee
By the inclusion of the additional term $\delta k(S,\bar{S})$, we have in mind incorporating non-perturbative contributions arising from Shenker-like effects, as mentioned in the Introduction and further described in~\cref{sec:npeffects}. For the superpotential we take the general form in~\cref{eq:superpot1}. We consider $T$-dependent K\"{a}hler potential corrections in~\cref{sec:npeffects}, and we return to the general case of~\cref{eq:modinvKahpot} in~\cref{sec:anomaly}.

Using~\cref{eq:kahpot1,eq:minikah,eq:superpot1}, the F-term supergravity scalar potential is
%=====================
\be
\begin{split}
    V(T,S) &= e^{\mathcal{K}}\left(\mc{K}^{S\bar{S}}F_S\bar{F}_{\bar{S}} \,+\,\mc{K}^{T\bar{T}}F_T\bar{F}_{\bar{T}} \,-\,3 W\bar{W}\right)\\
        &= e^{k(S,\sbar)}\,Z(T,\tbar)\,\lvert\Omega(S)\rvert^2\,\left[\left(A(S,\sbar)-3\right)\lvert H(T)\rvert^2+\widehat{V}(T,\tbar)\right]\coma
\end{split}
\label{eq:pot1} 
\ee
%=====================
where we defined
%=====================
\begin{align}
    A(S,\bar{S}) &= \frac{k^{S\sbar}F_S\bar{F}_{\sbar}}{\lvert W\rvert^2} =\frac{k^{S\sbar}\lvert\Omega_S + K_S\Omega\rvert^2}{\lvert\Omega\rvert^2}\coma\label{eq:potparts1}\\
    \widehat{V}(T,\tbar) &= \frac{-(T-\tbar)^2}{3}\,\bigg\lvert H_T(T) - \frac{3i}{2\pi}H(T)\widehat{G}_2(T,\tbar)  \bigg\rvert^2\coma\label{eq:potparts2}\\
    Z(T,\tbar) &= \frac{1}{i(T-\tbar)^3 \lvert\eta(T)\rvert^{12}}\coma
    \label{eq:potparts3}
\end{align}
%=====================
with subscripts denoting derivatives and $k^{S\bar{S}} = (k_{S\bar{S}})^{-1}$. We have also introduced the non-holomorphic Eisenstein series of weight $(2,0)$, $\widehat{G}_2(T,\bar{T})$ (see \cref{app:modularforms}).
Note that each of the functions defined in~\cref{eq:potparts1,eq:potparts2,eq:potparts3} are modular invariant. 

This potential has been discussed in several contexts in the literature. Originally in~\cite{Font:1990nt, Cvetic:1991qm} as an effective field theory for heterotic phenomenology, then in the context of relating swampland conjectures to modular symmetry in~\cite{Gonzalo:2018guu}, and recently in a study to connect flavor symmetry and modular symmetry~\cite{Novichkov:2022wvg}. A fascinating feature of this potential is that it diverges in the limit $\text{Im}(T)\rightarrow\infty$. This has interesting implications for the Swampland Distance Conjecture~\cite{Ooguri:2006in}, as observed in~\cite{Gonzalo:2018guu,Cribiori:2022sxf}. In the following, our discussion will align mostly with the first two contexts -- we will study vacua of the above potential with Shenker-like terms included via 
$k(S,\bar{S})$.

\subsection{Extrema of the Two-Modulus Model}

In this section, we study the vacua of the two-modulus heterotic model with the potential in \cref{eq:pot1}. Several aspects of the vacua have been studied in great detail in~\cite{Font:1990nt,Cvetic:1991qm, Gonzalo:2018guu}. We now review relevant details of those discussions.

An important feature of the scalar potential in \cref{eq:pot1} is that it is a modular function -- that is, the scalar potential is invariant under the $\pslz$ transformation defined in~\cref{eq:slztransf}. In addition to restricting the form of the non-perturbative superpotential, the power of modular symmetry also guides the search for vacua. Since $dT$ transforms as $dT\rightarrow (c T+d)^{-2}dT$ under modular transformations, $\partial V/\partial T$ is a weight $(2,0)$ non-holomorphic modular form. As shown in \cref{app:modularforms}, weight $(2,0)$ modular forms vanish at the $\pslz$ fixed points $T=i$ and $T=e^{2\pi i/3}\equiv \rho$. Thus
%=====================
\be
        \partial_T V(S,T)\rvert_{T=i,\rho} =0\coma
\ee
%=====================
and the fixed points are always extrema in the $T$-sector. Modular symmetry also simplifies the  analysis of the critical points at the fixed points -- the mixed derivatives of $S$ and $T$ are also weight $(2,0)$ modular forms, and so
%=====================
\be
    \partial_S\partial_T V(S,T)\rvert_{T=i,\rho} =\partial_{\bar{S}}\partial_T V(S,T)\rvert_{T=i,\rho} = 0\fstop
\ee
%=====================
The K\"{a}hler modulus sector can also have critical points away from the fixed points, but such points are more difficult to analyze since modular symmetry does not assist and their treatment must be purely numerical. Thus we can partially categorize extrema by the value of their K\"{a}hler modulus vev as follows:
%=====================
    \begin{align}
        &\text{Class 1: }\,\quad \hspace{0.01cm}\langle T\rangle =i\coma\\
        &\text{Class 2: }\,\quad \langle T\rangle =\rho\coma\\
        &\text{Class 3: }\,\quad \langle T\rangle \neq i,\rho\fstop
    \end{align}
%=====================
It was conjectured~\cite{Cvetic:1991qm} that all critical points lie on either the boundary of the fundamental domain of $T$ or the line $\text{Re}(T)=0$.  However,~\cite{Novichkov:2022wvg} disputes this conjecture by finding minima inside the fundamental domain and close to the fixed point $\rho$. We reinforce these results by finding multiple saddle points inside the fundamental domain, as discussed in \cref{sec:outerrim}.

To completely understand the vacua, we must also consider the dilaton sector. Our goal will be to examine dS vacua, which can only be achieved if one or both of the F-terms
\begin{align}
    F_S &=\frac{H(T)}{\eta^6(T)}(\Omega_S+k_S\Omega)\equiv \frac{H(T)}{\eta^6(T)}\widetilde{F}_S\coma\\
    F_T &= \frac{\Omega(S)}{\eta^6(T)}\bigg(H_T - \frac{3i}{2\pi}\widehat{G}_2 H\bigg)\equiv \frac{\Omega(S)}{\eta^6(T)}\widetilde{F}_T\coma
\end{align}
are non-zero. We have also introduced the re-scaled F-terms $\widetilde{F}_S$ and $\widetilde{F}_T$ for later convenience. We can then further categorize vacua according to whether or not they force the dilaton F-term to vanish:
%=====================
    \begin{align}
        &\text{Class A: }\,\quad \widetilde{F}_S=\Omega_S+k_S\Omega=0\coma\\
        &\text{Class B: }\,\quad \widetilde{F}_S \neq 0,\nonumber\\
                                &\hspace{2.1cm}\Omega_{SS} =\bar{\Omega}e^{2i\sigma}k_{S\bar{S}}\left(2-\frac{1}{\lvert H\rvert^2}\widehat{V}(T,\tbar)\right)+\left(\frac{k_{SS\bar{S}}}{k_{S\bar{S}}}-k_S\right)\widetilde{F}_S - k_{SS}\Omega -k_S\Omega_S\label{eq:classBcond}\fstop
    \end{align}
%=====================
Here we have defined $\sigma = \text{arg}(\Omega_S + k_S \Omega)$. Note that both of the above conditions solve $\partial_S V(S,T)=0$, and it is assumed that $H(T)$ is non-vanishing for Class B extrema.

The Class A extrema simply correspond to the vanishing of the dilaton F-term. Thus to have any hope of achieving a vacuum with positive energy, along this branch we demand $F_T\neq 0$. If we again consider the simple case of a single gaugino condensate, such an extremum would require a non-physical negative string coupling constant. This is the motivation behind racetrack models, where the Class A solution corresponds to stabilization of the dilaton by balancing two gaugino condensates against one another. With $H(T)=1$, a vacuum exists for Class A solutions at the point $T\simeq 1.23i$, which is the typical stabilized value found in heterotic models.  
Note that for Class A solutions, the Hessian is block diagonal, independent of the value of $T$. This point is discussed in more detail in the next subsection and is a crucial aspect of the no-go theorem we prove there.

For Class B solutions, we can allow $F_T$ to vanish since positive energy could be achieved by the dilaton sector. These solutions are somewhat more unfamiliar as they do not follow this simple picture arising from racetracks and one must introduce a means to stabilize the dilaton  and generate a non-zero F-term. This can be achieved by the Shenker-like terms in $k(S,\bar{S})$. The study of these extrema is the subject of \cref{sec:vacua} and mechanisms to generate these vacua are discussed in \cref{sec:npeffects}.

Using the above categories, we organize potential dS extrema of the two-modulus theory into 6 classes: Class A-1, A-2, and A-3 extrema, which are SUSY-preserving in the dilaton direction, and Class B-1, B-2, and B-3 extrema, which instead break SUSY in the dilaton sector. Further refinements of each class are possible since the type of critical point in general depends on the integers $m$ and $n$ and the polynomial $\mathcal{P}(j(T))$ in \cref{eq:Hpara}. 

The Class A-1, A-2, and A-3 extrema were examined in~\cite{Cvetic:1991qm,Gonzalo:2018guu}. If one assumes that the dilaton is stabilized, then the analysis of these extrema reduces to examination of the K\"{a}hler modulus sector. The authors of \cite{Cvetic:1991qm,Gonzalo:2018guu} prove that the fixed points are never dS minima -- that is, there are no dS minima in Class A-1 and A-2 extrema. As for Class A-3,~\cite{Gonzalo:2018guu} argues that numerically they do not find dS minima and conjecture that they do not exist. We now verify this conjecture by proving a no-go theorem.

\subsection{Class A de Sitter No-Go Theorem} 
We now prove a no-go theorem that illustrates the impossibility of obtaining dS vacua via Class A solutions of the two-modulus model above. As a corollary, we will verify and extend the results of~\cite{Gonzalo:2018guu}. 

\clearpage

%\textbf{Theorem 1:} 
\begin{theorem}
 At a point $(T_0,S_0)$, the scalar potential $V(T,S)$ in~\cref{eq:pot1}\\ 
 \indent\hspace{1.8cm}can not simultaneously satisfy:
\begin{enumerate}[label=(\roman*).,leftmargin=4.8\parindent]
    \item $V(T_0,S_0) > 0$
    \item $\partial_SV(T_0,S_0)= 0 \quad \&  \quad \partial_TV(T_0,S_0) = 0 $
    \item $(\Omega_S+k_S\Omega)\rvert_{S=S_0}=0$
    \item Eigenvalues of the Hessian of $V(T,S)$ at $(T_0,S_0)$ are all $\ge 0$.\vspace{0.2cm}
\end{enumerate}
\end{theorem}

%\textit{Proof:}
\begin{proof}
The proof proceeds by contradiction -- let us assume that (i)-(iv) are true at $(T_0,S_0)$. The first derivative of $V(T,S)$ with respect to $S$ is
%=====================
\begin{equation}
    \partial_S V(T,S) = %\Omega(S)^{-1}\{\Omega_S(S) + K_S\Omega(S)\}V(T,S) 
    \frac{F_S}{W}V(T,S)+ \left\{e^{k(S,\Sb)} \lvert\Omega(S)\rvert^2\lvert H(T)\rvert^2Z(T,\tbar)\right\}\partial_S A(S,\Sb)\fstop
    \label{eq:VSder}
\end{equation}
%=====================
Since $F_S \propto \Omega_S +k_S \Omega$, (iii) manifestly implies the vanishing of the first term. The derivative $\partial_SA(S,\Sb)$ contains several terms, but each term is proportional to either $F_S$ or its conjugate, and so the second term above also vanishes. Thus (iii) implies the vanishing of the above derivative at $(T_0,S_0)$. This is simply a verification that (iii) defines Class A extremum. To be consistent with (i), we require that $\widetilde{F}_T(T_0)\neq 0$ and $\Omega(S_0)\neq 0$. Then without loss of generality we can introduce a parameter $\Lambda>0$ and recast (i) as 
\be
V(T_0,S_0) = e^{k_0}\lvert\Omega_0\rvert^2Z_0\Lambda^4\coma\label{eq:firstreq}
\ee
where the subscript denotes evaluation at $(T,S)= (T_0,S_0)$. This can be interpreted as an equation for $H_T(T_0)$ and is solved by 
%=====================
\be
H_T(T_0) = \frac{3i}{2\pi}H_0\widehat{G}_2(T_0,\tbar_0) \pm \frac{\sqrt{3}i}{T_0-\tbar_0}\bigg(\Lambda^2\! \pm i\sqrt{\lvert H_0\rvert^2\left(3-A(S_0,\bar{S}_0)\right)}\bigg)\fstop
\label{eq:Hprime1}
\ee
%=====================
In this expression, any of the four sign combinations is valid, and we have taken $A(S_0,\bar{S}_0)<3$. This is not an assumption since technically (iii) implies that $A(S_0,\bar{S}_0)=0$, but we will carry it in our expressions for the moment and take the appropriate limit at the end. Similarly, the $\partial_T V(T_0,S_0)$ requirement of condition (ii) yields an algebraic equation for $H_{TT}(T_0)$ that can be solved. We give this condition in~\cref{app:HTT}. Moving onto condition (iv), we first note that~\cref{eq:VSder} and (iii) imply that  
%=====================
\be
\partial_T^k\partial_{\tbar}^l\partial_S^{} V(T_0,S_0)=0 \quad\quad \forall\,\, k,\,l\in\mathbb{N}_+\fstop
\ee
%=====================
Thus the Hessian of $V(T,S)$ is block diagonal, with the blocks corresponding to the $T$ and $S$ sectors. Then the eigenvalues of the Hessian are simply the eigenvalues of the two blocks. We focus on the K\"{a}hler modulus block -- in terms of the real and imaginary components of $T=a+it$, its components are 
%=====================
\begin{align}
    \partial^2_tV &= 2\partial_T \partial_{\tbar}V -2\text{Re}(\partial_T^2 V)\coma\\
    \partial_a^2V &=2\partial_T \partial_{\tbar}V +2\text{Re}(\partial_T^2 V)\coma\\
     \partial_t\partial_a V &= -2 \text{Im}(\partial_T^2V) \fstop
\end{align}
%=====================
Computing the $T$ derivatives and plugging in the expressions derived above for $H_T(T_0)$ and $H_{TT}(T_0)$, we find
%=====================
\begin{align}
    \partial_T\partial_{\tbar} V(T_0,S_0) &= \frac{2e^{k_0}\lvert\Omega_0\rvert^2Z_0}{(T_0-\bar{T}_0)^2}\bigg(\Lambda^4(2-3A(S_0,\bar{S}_0))-\lvert H_0\rvert^2A(S_0,\bar{S}_0)\bigg)\\
   &\longrightarrow  \frac{-2e^{k_0}\lvert\Omega_0\rvert^2Z_0}{-(T_0-\bar{T}_0)^2} \Lambda^4\fstop
\end{align}
%=====================
In going to the second line, we have enforced the vanishing of $A_0$ as demanded by (iii), the Class A extremum condition. We see that $\partial_T\partial_{\tbar}V(T_0,S_0)<0$ for all values of $T_0$ and functions $H(T)$. Thus, it must be that $\partial_t^2V(T_0,S_0)<0$ or $\partial_a^2V(T_0,S_0)<0$. In either case, the determinant of the K\"{a}hler modulus block of the Hessian is negative. This immediately implies that one of the eigenvalues of the Hessian is negative, in contradiction with (iv). We can also consider the case of a vanishing $T$-block determinant, which could occur if
%=====================
\be
        \re\!(\partial_T^2 V(T_0,S_0)) = \pm\frac{2e^{k_0}\lvert\Omega_0\rvert^2Z_0}{(T_0-\bar{T}_0)^2} \Lambda^4 \;\; \quad \& \;\;\quad  \im\!(\partial_T^2 V(T_0,S_0))=0 \fstop
\ee
%=====================
However, the vanishing of the K\"{a}hler sector determinant for this case occurs only because one of its two eigenvalues vanishes. Even in this degenerate case, the other eigenvalue is necessarily nonzero and negative, again in contradiction with (iv). Thus (iv) is incompatible with (i)-(iii) and the theorem is demonstrated. Since conditions (i)-(iv) are the requirements for a Class A dS vacuum, the theorem demonstrates that such vacua are impossible. \end{proof}

\noindent The above result deserves several comments. First, the above is a proof that dS vacua cannot occur anywhere in the fundamental domain of $T$ if $\widetilde{F}_S=0$. This is equivalent to the statement that no non-perturbative superpotential in the class defined by~\cref{eq:superpot1} and~\cref{eq:Hpara} can lift a racetrack vacuum to positive energy. This is true even if Shenker-like effects are present in the K\"{a}hler potential, assuming they enter only through $k(S,\bar{S})$. This is a limited extension of previous no-go results mentioned in the introduction to non-perturbative corrections of strength $\mathcal{O}(e^{-1/g_s})$. Second, we immediately have two straightforward corollaries: 

\vspace{0.2cm}

%\textbf{Corollary 1.1:} 
\begin{corollary}
Class A extrema in the two-modulus model with\\ 
\indent \hspace{2.3cm}$k(S,\bar{S})=-\ln(S+\bar{S})$ can never be dS vacua.
\end{corollary}

%\textbf{Corollary 1.2:} 
\begin{corollary}
The one-modulus model with $W(T) = H(T)/\eta^6(T)$ and\\ \indent\hspace{2.3cm}$\mc{K} = -3\ln(-i(T-\tbar))$ can not have dS vacua. \vspace{0.2cm}
\end{corollary}

\noindent The latter follows by taking the appropriate limits to remove the dilaton in the proof of Theorem 1. These corollaries verify the conjectures in~\cite{Gonzalo:2018guu} that state that neither Class A extrema nor the single-modulus model can have dS vacua. %\nr{Although we are not consideing a concrete model with full moduli stabilization, this also explains why \cite{Parameswaran:2010ec} finds only unstable dS extrema.} 

We also observe that the above proof does not make use of the modular properties of $H(T)$ in any meaningful way. We posit that the lack of Class A dS vacua is more tied to the structure of Class A extrema and the factorized form of the superpotential and K\"{a}hler potential.

Next, we note that if we were to assume that $(T_0,S_0)$ was an AdS minimum by replacing $\Lambda^4 \rightarrow -\Lambda^4$, then the no-go theorem does not apply because manifestly $\partial_T\partial_{\tbar}V(T_0,S_0)>0$ -- the very argument that forbids dS minima cleanly allows for AdS minima. 

Finally, the above no-go provides a hint as to how one might obtain dS vacua in these heterotic compactifications -- one must consider Class B solutions where the dilaton F-term is non-vanishing. In the next section, we explore constraints on Class B extrema such that dS vacua exist in the two-modulus model.

\section{Circumventing the No-Go}
\label{sec:vacua}
In this section, we study the branch of extrema that evade the no-go result of the previous section -- namely, Class B extrema. We will assume that the dilaton is stabilized and determine under what conditions the K\"{a}hler modulus sector is stabilized with positive energy. This will translate into bounds on the function $A(S,\bar{S})$, which depends on the Shenker-like effects via $k(S,\bar{S})$. For Class B-1 and B-2 extrema, which occur at the fixed points, this analysis is valid due to the block diagonal structure of the Hessian required by modular invariance, as discussed in~\cref{sec:heterotic}. At general points in the fundamental domain, corresponding to Class B-3 extrema, the Hessian is in general not block diagonal and one must treat the $T$ and $S$ moduli together. For these extrema, we will examine minima of the single-modulus model containing only $T$ and discuss the plausibility of uplifting them via the dilaton subsector. In doing so, we will present evidence disproving a previous conjecture on the nature of extrema in the fundamental domain of $T$. We will turn to the issue of stabilizing the dilaton sector in the following section.
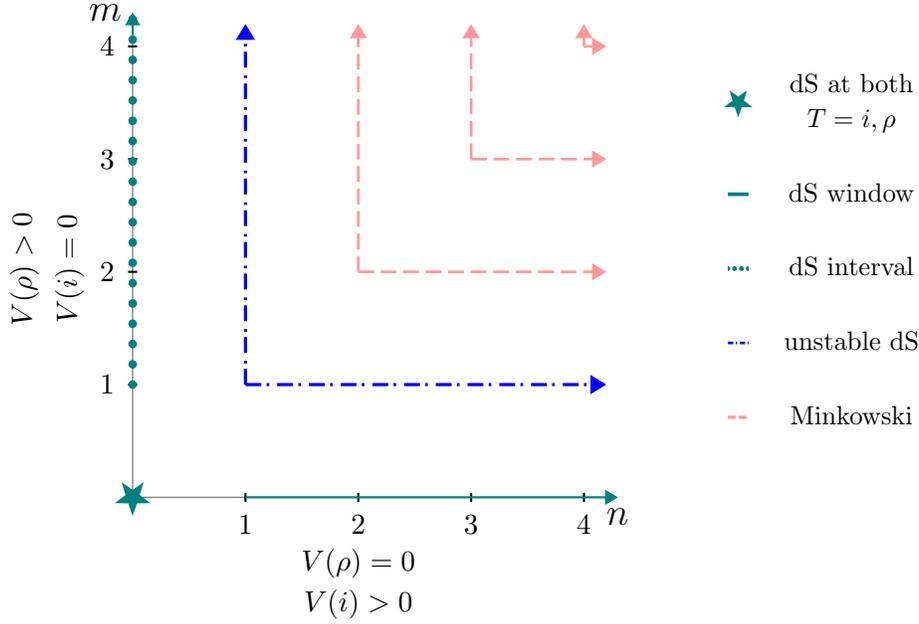
\begin{figure}
    \centering
\begin{tikzpicture}[scale=1.5]
\draw[line width=0.5pt,gray,-Triangle] (0,0) -- (4.3,0) node[below] {\color{black}{\large{$n$}}};
\draw[line width=0.5pt,gray,-Triangle] (0,0) -- (0,4.3) node[left] {\color{black}{\large{$m$}}};

%\draw[line width=0.3pt,dashed] (1,0) -- (1,1);
\draw[line width=1.2pt,blue,dash pattern={on 6pt off 2pt on 1pt off 3pt},-Triangle] (1,1) -- (1,4.2);
%\draw[line width=0.3pt,dashed] (2,0) -- (2,2.8);
%\draw[line width=0.3pt,dashed] (0,2) -- (2.8,2);
%\draw[line width=0.3pt,dashed] (0,1) -- (1,1);
\draw[line width=1.2pt,blue,dash pattern={on 6pt off 2pt on 1pt off 3pt},-Triangle] (1,1) -- (4.2,1);
\draw[thick] (1,-0.04) -- (1,0.04) node[below,xshift=0cm,yshift=-0.15cm] {\color{black}{\small{$1$}}};
\draw[thick] (2,-0.04) -- (2,0.04) node[below,xshift=0cm,yshift=-0.15cm] {\color{black}{\small{$2$}}};
\draw[thick] (3,-0.04) -- (3,0.04) node[below,xshift=0cm,yshift=-0.15cm] {\color{black}{\small{$3$}}};
\draw[thick] (4,-0.04) -- (4,0.04) node[below,xshift=0cm,yshift=-0.15cm] {\color{black}{\small{$4$}}};
\draw[thick] (-0.04,1) -- (0.04,1) node[left,xshift=-0.15cm,yshift=0cm] {\color{black}{\small{$1$}}};
\draw[thick] (-0.04,2) -- (0.04,2) node[left,xshift=-0.15cm,yshift=0cm] {\color{black}{\small{$2$}}};
\draw[thick] (-0.04,3) -- (0.04,3) node[left,xshift=-0.15cm,yshift=0cm] {\color{black}{\small{$3$}}};
\draw[thick] (-0.04,4) -- (0.04,4) node[left,xshift=-0.15cm,yshift=0cm] {\color{black}{\small{$4$}}};
% \node at (1.9,1.8){\footnotesize\newshortstack{Into the\\fundamental\\domain}};

\node at (2,0)[inner sep=1pt,label={[xshift=0cm,yshift=-1.25cm]:\small{$V\!\left(\rho\right)=0$}}]{};
\node at (2,0)[inner sep=1pt,label={[xshift=0cm,yshift=-1.8cm]:\small{$V\!\left(i\right)>0$}}]{};
\node at (0,2)[inner sep=1pt,label={[xshift=-0.5cm,yshift=0cm,rotate=90]:\small{$V\!\left(i\right)=0$}}]{};
\node at (0,2)[inner sep=1pt,label={[xshift=-1.1cm,yshift=0cm,rotate=90]:\small{$V\!\left(\rho\right)>0$}}]{};

\node at (0,0)[star,star points=5,star point ratio=0.4,fill=teal,inner sep=5pt,label={[xshift=0cm,yshift=-0.7cm,teal]:}]{};

%\node at (1,1)[circle,fill=blue,inner sep=2pt,label={[xshift=0.4cm,yshift=-0.7cm,teal]:}]{};

%\node at (2,2)[label={[red]:\small{Inferno}}]{};
%\fill[red!20,nearly transparent,path fading=east] (2,2) -- (2,2) -- (4.1,2) -- (4.1,4.1) -- (2,4.1);

\draw[line width=1pt,red!40,-Triangle,dash pattern={on 6pt off 2.5pt}] (2,2) -- (4.2,2);
\draw[line width=1pt,red!40,-Triangle,dash pattern={on 6pt off 2.5pt}] (2,2) -- (2,4.2);
\draw[line width=1pt,red!40,-Triangle,dash pattern={on 6pt off 2.5pt}] (3,3) -- (4.2,3);
\draw[line width=1pt,red!40,-Triangle,dash pattern={on 6pt off 2.5pt}] (3,3) -- (3,4.2);
\draw[line width=1pt,red!40,-Triangle,dash pattern={on 6pt off 2.5pt}] (4,4) -- (4.2,4);
\draw[line width=1pt,red!40,-Triangle,dash pattern={on 6pt off 2.5pt}] (4,4) -- (4,4.2);

%\draw[line width=1pt,red!40,nearly transparent,-Triangle] (2.01,2.8) -- (2.01,3.03); 
%\draw[line width=1pt,red!40,nearly transparent,-Triangle] (2.8,2.01) -- (3.03,2.01);

\draw[line width=0.8pt,teal,-Triangle] (1,0) -- (4.3,0);
\draw[line width=0.8pt,teal,-Triangle] (0,4.1) -- (0,4.3) ;
\draw[decorate sep={1.mm}{2.7mm},teal,fill=teal] (0,1) -- (0,4.3) ;

%%%%%%%%%%
\matrix [draw=white,below left,nodes={anchor=center},column sep = 0.07cm, row sep = 0.25cm] at (7.3,4) {
  \node [star,star points=5,star point ratio=0.4,fill=teal,inner sep=4pt] {};& \node{\footnotesize\newshortstack{dS at both\\ $T=i,\rho$}};\\
  %\node [circle,fill=blue,inner sep=2pt] {};& \node{\footnotesize\newshortstack{unstable dS}};\\
  \node[path picture={
    \draw[teal,thick,line width=1.4pt](path picture bounding box.west) -- (path picture bounding box.east);
   }
   ]{}; 
   & \node{\footnotesize\newshortstack{dS window}};\\
   \node[path picture={
    \draw[decorate sep={0.5mm}{0.9mm},teal,fill=teal] (path picture bounding box.west) -- (path picture bounding box.east);
   }]{}; & \node{\footnotesize\newshortstack{dS interval}};\\
   \node[path picture={
    \draw[blue,thick,line width=1pt,dash pattern={on 2pt off 1pt on 0.5pt off 1pt}] (path picture bounding box.west) -- (path picture bounding box.east);}]{}; & \node{\footnotesize\newshortstack{unstable dS}};\\
   \node[path picture={
    \draw[red!40,thick,line width=1.2pt,dash pattern={on 3pt off 1pt}] (path picture bounding box.west) -- (path picture bounding box.east);
   }]{}; & \node{\footnotesize\newshortstack{Minkowski}};\\
};
\end{tikzpicture}
\caption{Summary of dS vacua possibilities at the fixed points with respect to the integers $(m,n)$ in $H(T)$ assuming the dilaton is stabilized. At $(0,0)$, both fixed points can have a dS vacuum. For $(m,0)$, $T=\rho$ is a dS minimum if $A(S,\bar{S})>3$. For $(0,n)$, $T=i$ can be a dS vacuum for a window of $A(S,\bar{S})$ values that increases with $n$ and is dependent on $\mathcal{P}(j(T))$. The cases $(1,n)$ or $(m,1)$ result in unstable dS extrema. Finally, we always have Minkowski extrema at the self dual points when $(m,n)>(1,1)$. }
\label{fig:divinecomedy}
\end{figure}

\subsection{Class B-1 Vacua: $T=i$}
We start by examining extrema at the fixed point $T=i$. 
If the integer $m$ in \cref{eq:superpot0} is greater than one, then both $H(T)$ and $\widehat{V}(T,\tbar)$ vanish at $T=i$ and so the extremum is Minkowski.\footnote{This follows from the properties of the Eisenstein functions $G_4$ and $G_6$ entering the definition of $H(T)$: $G_4(\rho)=G_6(i)=0$ while $G_4(i)$ and $G_6(\rho)$ are non-zero.} Therefore, in the remainder of this section we will study the potential for the cases $m=0$ and $m=1$.

\paragraph{m = 0 :}
The potential evaluated at $T=i$ has the compact form 
\be
  V(S,\sbar,i,-i) = \frac{2^{4n+9}\pi^{8n+9}}{3^{2n}5^{2n}\Gamma^{12}(1/4)}\lvert\Omega(S)\rvert^2\lvert\mc{P}(1728)\rvert^2 e^{k(S,\sbar)}\left(A(S,\Sb)-3\right)\coma
\ee
and we see then that a dS extremum requires $A(S,\Sb)>3$. To examine the stability of the extremum, we calculate the second order derivatives in the fields $t,\,a$ from $T=a+it$ as
\be\begin{split}
\partial_{t/a}^2V\rvert_{T=i}= &\frac{4^{2n+4}\pi^{8 n+9} e^{k(S,\Sb)}}{3^{2 n-1}5^{2n}\Gamma^{12}(1/4)}\lvert\Omega(S)\rvert^2\lvert\mc{P}(1728)\rvert^2\bigg[(A(S,\Sb)-2)\,\\
&\left.\pm(A(S,\Sb)-1)\frac{\Gamma^8(1/4)}{192\pi^4}\left(\!1+8n+41472\,\re\!\left(\!\frac{\mc{P}^\prime(1728)}{\mc{P}(1728)}\!\right)\!\right)\!+\lvert\mc{B}_n\rvert^2\right]\,,
\end{split}
\ee
where the plus sign refers to the derivatives in $t$ and the minus to the ones in $a$ and we have defined
\be\label{eq:definitionB}
\mc{B}_n\equiv\frac{\Gamma^8(1/4)}{192\pi^4}\left(1+8n+41472\,\frac{\mc{P}^\prime(1728)}{\mc{P}(1728)}\right)\coma
\ee
and the mixed derivative reads
\be
\partial_{t,a}^2V\rvert_{T=i}= \frac{4^{9+2n }\pi^{5+8n}2\,e^{k(S,\Sb)}}{3^{2n}5^{2n}\Gamma^4(-3/4)}\lvert\Omega(S)\rvert^2\lvert\mc{P}(1728)\rvert^2\!\left(A(S,\Sb)-1\right)\im\!\!\left(\frac{\mc{P}^\prime(1728)}{\mc{P}(1728)}\right)\fstop
\ee
The general conditions on $A(S,\bar{S})$ are complicated and involve a number of subcases. We describe them in detail in~\cref{app:icases}. As a concrete example, we set $\mathcal{P}(j(T))=1$. The condition for a minimum corresponds to the case $\mc{B}_n>1$ found in~\cref{app:icases} and takes the form 
\be
    2-\frac{(1+8n)\Gamma^8(1/4)}{192\pi^4}<A(S,\bar{S})<2+\frac{(1+8n)\Gamma^8(1/4)}{192\pi^4}\fstop
\ee
If we set $n=0$ -- the trivial case of $H(T)=1$ -- we find a range for $A(S,\bar{S})$:
\begin{align}\label{eq:AP1n0}
    0.4035<A(S,\sbar)<3.5964 \fstop
\end{align}
Thus there is a narrow window in which it is possible to have at the point $T=i$ a dS minimum. However, we note several interesting features. If $n>0$, then the left boundary of stability goes to negative $A(S,\sbar)$ and so the extremum is a minimum for all positive values of $A(S,\bar{S})$. Furthermore, as $n$ increases, the window for dS minima grows. We display examples of this window of stability for $\mathcal{P}(j(T))=1$ and $\mathcal{P}(j(T))=1+j(T)$ in~\cref{fig:ClassB1Hess}. 

\begin{figure}
 \centering 
 \includegraphics[width=\textwidth,keepaspectratio]{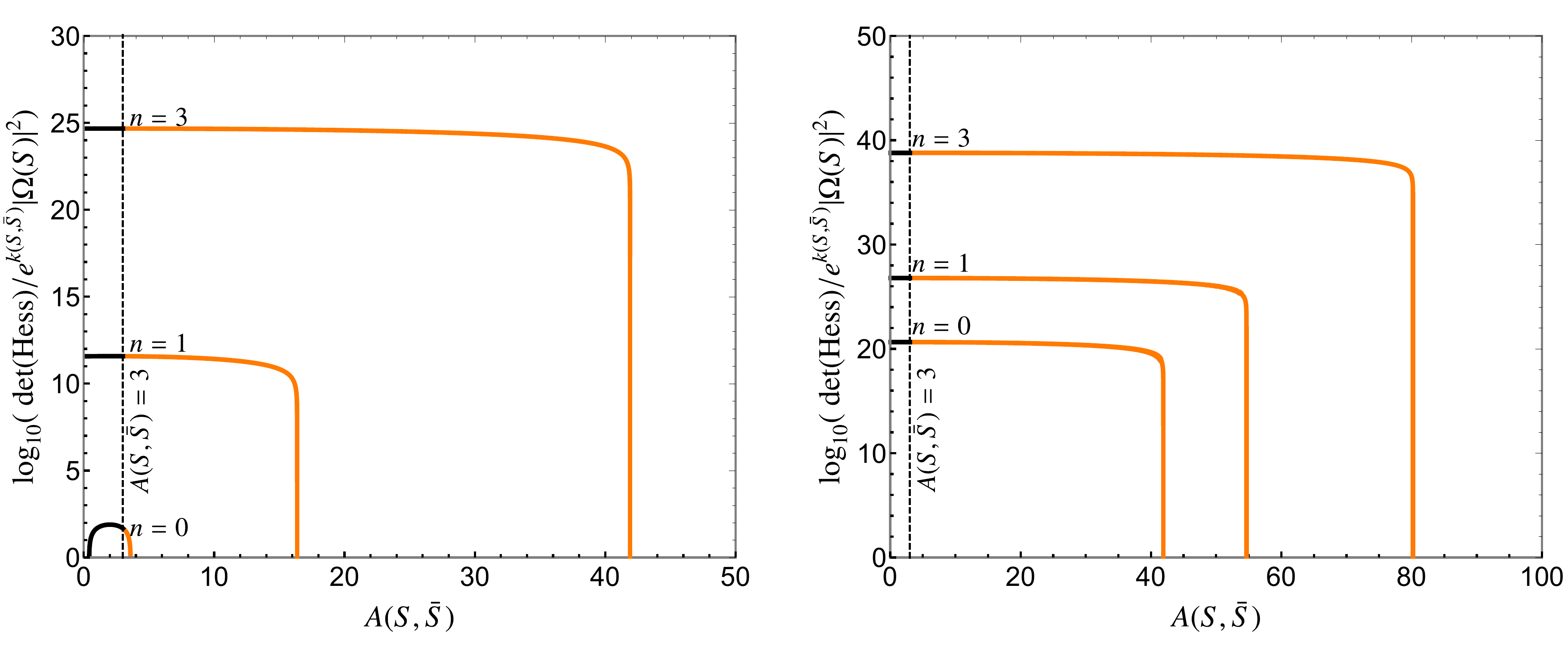}
 \caption{Determinants of the K\"{a}hler modulus block of the Hessian at $T=i$ with $(m,n)=(0,n)$ and polynomials given by $\mathcal{P}(j(T))=1$ (left) and $\mathcal{P}(j(T))=1+j(T)$ (right). The orange segments correspond to dS vacua assuming a stabilized dilaton. The vertical asymptote indicates the point where the determinant crosses to negative values.}
\label{fig:ClassB1Hess}
\end{figure}

\paragraph{m = 1:}
This value for $m$ is particularly intriguing because the potential at $T=i$ is positive for $\Omega(S)\neq 0$ and does not depend on $A(S,\Sb)$:
\be
 V(S,\sbar,i,-i)=\frac{2^{4 n+13}\pi^{8n+17}e^{k(S,\sbar)}}{49\times3^{2n+3}5^{2n+2}\Gamma^4(1/4)}\lvert \Omega(S)\rvert^2\lvert \mc{P}(1728)\rvert^2\fstop
\ee
The second derivatives read
\be
\begin{split}
\partial_{t/a}^2V\rvert_{T=i}=&\frac{2^{4n+12}\pi^{8n+17}e^{k(S,\sbar)}}{5^{2n+2}3^{2n+2} 49 \Gamma^4(1/4) }\lvert \Omega(S)\rvert^2\lvert\mc{P}(1728)\rvert^2\\
&\times\left[3 A(S,\Sb)-2\pm\frac{\Gamma^8(1/4)}{192\pi^4}\!\left(\!55+72n+373248\,\re\!\!\left(\!\frac{\mc{P}^\prime(1728)}{\mc{P}(1728)}\!\right)\!\right)\!\right]\,,
\end{split}
\ee
\be
\partial_{t,a}^2V\rvert_{T=i}=\frac{2^{4n+15}\pi^{8n+13}e^{k(S,\sbar)}\Gamma^4(1/4)}{ 5^{2n+2}3^{2n-2} 49}\lvert \Omega(S)\rvert^2\lvert\mc{P}(1728)\rvert^2\im\!\!\left(\frac{\mc{P}^\prime(1728)}{\mc{P}(1728)}\right)\coma
\ee
and requiring the eigenvalues of the Hessian to be positive leads to the condition
\be\label{eq:AforTim1}
    A(S,\Sb)>\frac{2}{3}+\frac{\Gamma^8(1/4)}{576 \pi^4}\left[55+72 n+373248\bigg\lvert\frac{\mc{P}^\prime(1728)}{\mc{P}(1728)}\bigg\rvert\right]\fstop
\ee
However, the dilaton derivative is
\be
\partial_S V\rvert_{T=i} = \frac{2^{4 n+13} \pi ^{8 n+17} \lvert \mathcal{P}(1728)\rvert^2}{49\times 3^{2 n+3}5^{2n+2} \Gamma ^4(1/4)} e^{k(S,\Sb)}\bar{\Omega}(\bar{S})(\Omega_S +k_S\Omega )\fstop
\ee
Discarding the $\Omega(S) = 0$ Minkowski solution, the above indicates that stabilizing the dilaton actually forces us into a Class A extremum. Hence the case $m=1$ is never a dS vacuum at $T=i$, in agreement with the tree-level K\"{a}hler potential analysis in~\cite{Gonzalo:2018guu}.

%========================================
\subsection{Class B-2 Vacua: $T=\rho$}
\label{sec:typeIIvacua}
%========================================
For $T=\rho$, we must set $n<2$ or else all extrema will be Minkowski due to the vanishing of $G_4(\rho)$. We again separate the discussion for the case of $n=0$ and $n=1$.
\begin{figure}
\centering
    \includegraphics[width=0.6\textwidth]{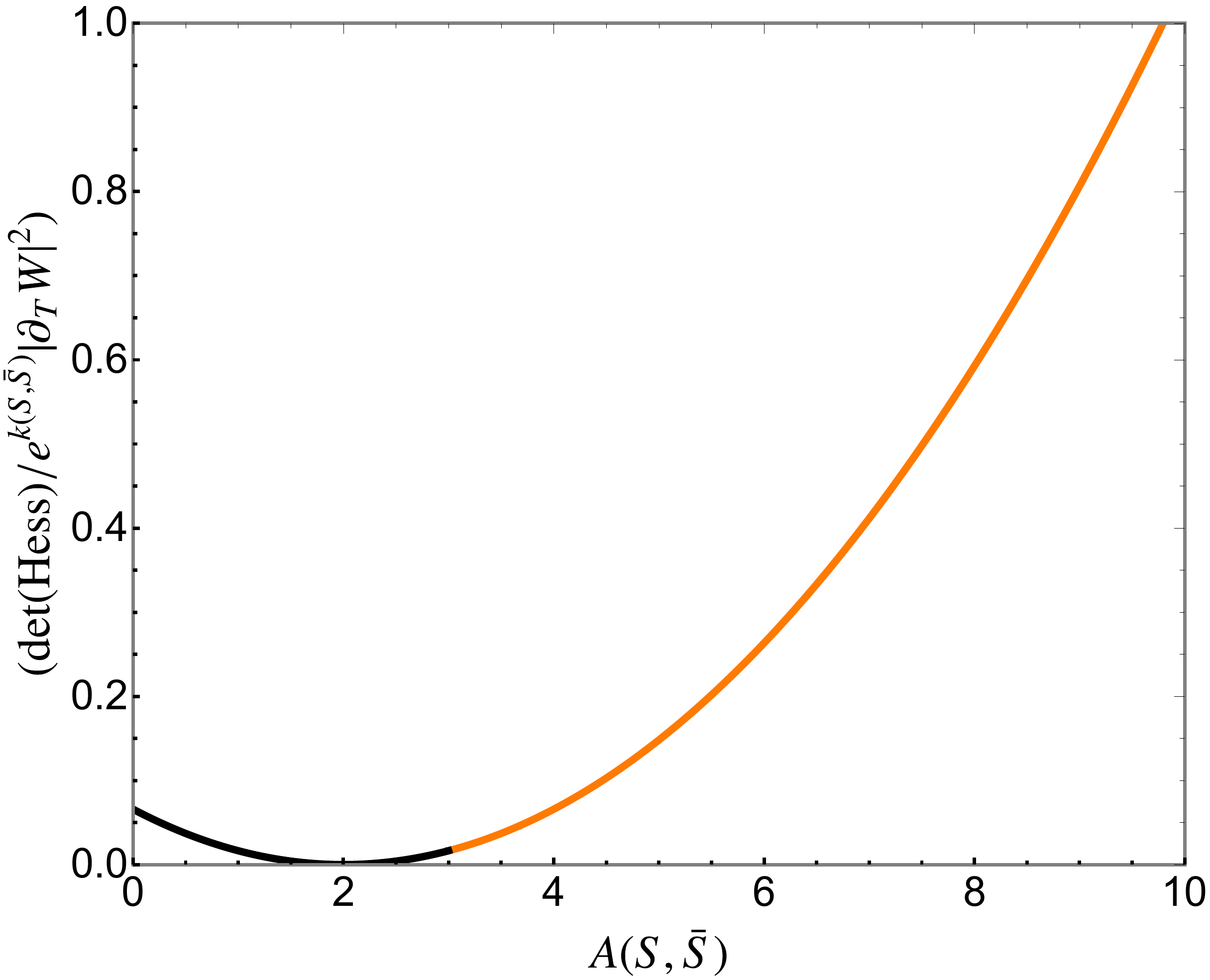}
  \caption{Rescaled Hessian determinant of the K\"{a}hler modulus block at $T=\rho$. Again the orange segments correspond to dS vacua assuming a stabilized dilaton.} 
  \label{fig:dethessianTrho}
\end{figure}
\paragraph{n = 0 :}
The value of the potential is then 
\be
  V(S,\sbar,\rho,\rho^*) = \frac{4^{4m+6}\pi^{12m+12}e^{k(S,\sbar)}}{ 3^{3m+3}1225^m\Gamma^{18}(1/3)}\lvert\Omega(S)\rvert^2\lvert\mathcal{P}(0)\rvert^2 \left(A(S,\Sb)-3\right)\coma
\ee
and the non-zero second order derivatives are
\be
\partial_t^2V\rvert_{T=\rho}= \partial_a^2V\rvert_{T=\rho} = \frac{2^{8m+13}\pi^{12m+12}e^{k(S,\Sb)}}{3^{3m+3} 1225^{m}\Gamma^{18}(1/3)}\lvert\Omega(S)\rvert^2\lvert\mathcal{P}(0)\rvert^2\left(A(S,\Sb)-2\right)\coma
\ee
so that the Hessian in the $T$ sector is 
\be\label{eq:detTrho}
     \text{det}(\text{Hess})\rvert_{T=\rho} = (\partial_t^2V)^2\rvert_{T=\rho} \fstop
\ee
Thus, we see that the determinant of the Hessian is always positive (assuming the dilaton is stabilized), and the eigenvalues are positive for $A(S,\Sb)>2$. Hence, there is a dS minimum for $A(S,\Sb)>3$. This is depicted in~\cref{fig:dethessianTrho}.

\paragraph{n = 1 :}
As in the previous section, for this value of $n$ the potential is always positive and does not depend on $A(S,\Sb)$:
\be
  V(S,\sbar,\rho,\rho^*) = \frac{4^{4m+6}\pi^{12m+14}e^{k(S,\sbar)}}{ 3^{3m+5}5^{2m+2}49^m\Gamma^{6}(1/3)}\lvert\Omega(S)\rvert^2\lvert\mathcal{P}(0)\rvert^2 \fstop
\label{eq:dilquint}
\ee
Also as above, $\partial_S V\rvert_{T=\rho}\propto \Omega_S+k_S\Omega$ and so it is not possible to stabilize the dilaton with $A(S,\bar{S})>0$. On the other hand, the determinant of the K\"{a}hler modulus block of the Hessian is 
\be
\partial_t^2V\rvert_{T=\rho}= \partial_a^2V\rvert_{T=\rho} = \frac{2^{8m+7}\pi^{12m+14}e^{k(S,\Sb)}}{5^{2m+2} 1323^{m}\Gamma^{6}(-2/3)}\lvert\Omega(S)\rvert^2\lvert\mathcal{P}(0)\rvert^2\left(3 A(S,\Sb)-2\right)\fstop
\ee
So long as $A(S,\bar{S}) >2/3$, the K\"{a}hler modulus sector is stabilized. This scenario cannot give dS vacua, but it does yield an intriguing model for dilaton quintessence: if the dilaton has an initial field value such that $A(S,\bar{S}) >2/3$, then as the dilaton runs towards weak coupling, the cosmological constant in~\cref{eq:dilquint} decreases but the compact dimensions are stabilized. Indeed if we assume that $\Omega(S)\sim \sum_a e^{-S/b_a}$  and the dilaton K\"{a}hler potential is $k(S,\bar{S}) = -\ln(S+\bar{S}) + \mathcal{O}(e^{-\sqrt{S}})$, then $A(S,\bar{S})$ grows with $S$ and the stability of $T$ is maintained. We leave a detailed study of the phenomenology and cosmology of this scenario to future studies. 

%========================================
\subsection{Class B-3 Vacua: Into the Fundamental Domain}%Outer Rim}
\label{sec:outerrim}
%========================================
We now step away from the fixed points and proceed elsewhere in the fundamental domain of $\pslz$. As mentioned above, the stability of Class B extrema away from the fixed points are in general difficult to analyze since the Hessian is not required to be block diagonal and a thorough analysis requires a full treatment of the dilaton sector. Our focus here will be a study on the nature of extrema in the K\"{a}hler modulus sector in the absence of a dilaton. We will comment on re-introduction of the dilaton at the end of the section.

It was previously conjectured~\cite{Cvetic:1991qm} that all extrema of the $T$ sector must lie on the union of the boundary points included in $\widetilde{\mathcal{F}}$ and the line $\text{Re}(T)=0$. The argument for this is that the scalar potential in~\cref{eq:pot1} displays an additional $\mathbb{Z}_2$ symmetry by swapping $T\leftrightarrow - \bar{T}$. The line $\text{Re}(T)=0$ consists of fixed points under this $\mathbb{Z}_2$, while the other boundary points in $\widetilde{\mathcal{F}}$ are fixed by the combined actions of $\mathbb{Z}_2$ and $\pslz$ transformations. Assuming smooth behavior on the union of these curves, the scalar potential is then extremized in at least one direction on the union of these curves. The conjecture utilizes this observation to posit that all extrema of the $T$ sector lie on the above union of curves.
%================
 \begin{figure}[t]
    \centering
\includegraphics[width=0.75\textwidth,keepaspectratio]{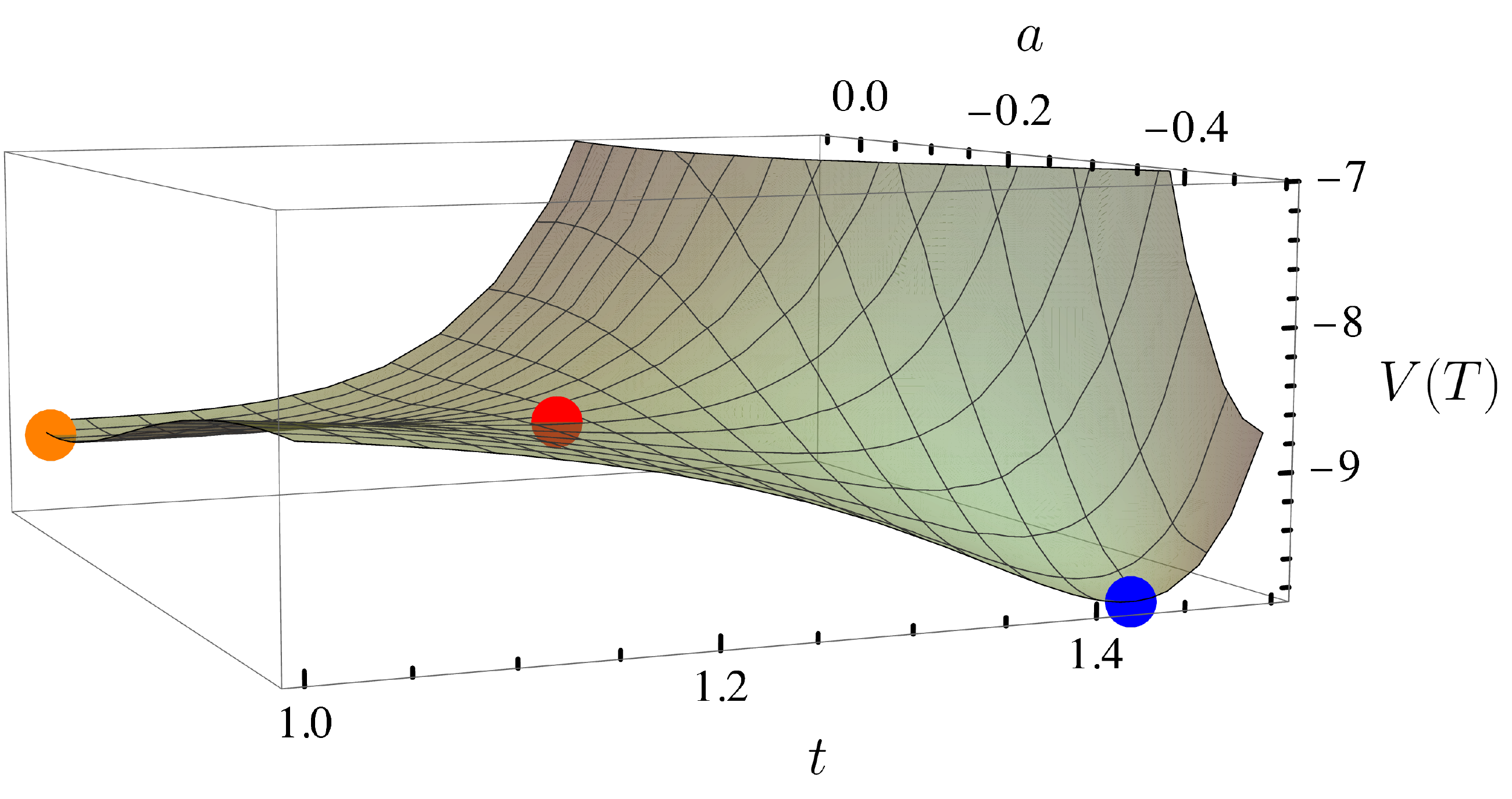}
\caption{Potential with two minima and a saddle interpolating between them. The orange dot is a minimum at $T=i$, the blue dot is a minimum on the line $\text{Re}(T)= -\frac{1}{2}$, and the red dot signals the saddle point inside the fundamental domain.}
    \label{fig:saddle_minima}
\end{figure}
%================ 

We now argue that this is not the case. We will restrict ourselves to a simple parametrization of $H(T)$:
\be
  H_1(T)=\bigg(\frac{G_4(T)}{\eta^8(T)}\bigg)^n\bigg(\frac{G_6(T)}{\eta^{12}(T)}\bigg)^m(1+\beta \,j(T))\fstop
\label{eq:H1}  
\ee
We then extremize~\cref{eq:pot1} with respect to $T$ and set $A(S,\bar{S})=0$. By scanning different values of $(m,n)\leq(5,5)$ and $\beta$, we found extrema for $\beta< 10^{-2}$ throughout the fundamental domain. First, we find new saddle point extrema inside the fundamental domain. This can be expected, as once one has a minimum on the boundary and another minimum at one of the fixed points, an extremum must exist which interpolates between these points. We show an example of such situation in \cref{fig:saddle_minima} where we have taken $(m,n)= (0,0)$ and $\beta = 10^{-5}$ in~\cref{eq:H1}.
Second, we confirm the results of~\cite{Novichkov:2022wvg}: close to the fixed point $T=\rho$, and for $\mc{P}(j(T))=1$, there exists a minimum which tends to $T=\rho$ at increasing values of $m$. In \cref{fig:fund_dom_minima} we show the first three concentric minima. Thus it appears that the conjecture in~\cite{Cvetic:1991qm} does not hold in general. However, this hints at the exciting possibility of a larger-than-anticipated \textit{modular landscape} of heterotic vacua. We leave a thorough search for further vacua in the fundamental domain using general $H(T)$ to future studies.

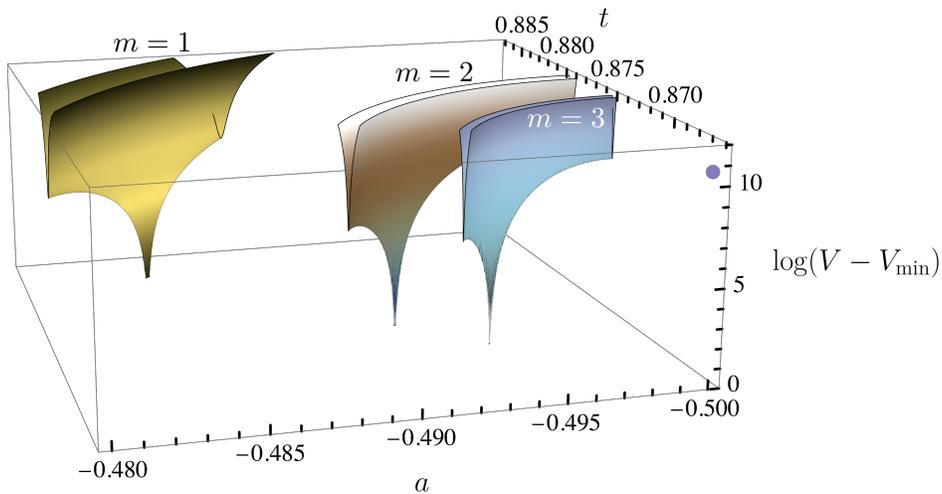
\begin{figure}[t]
    \centering
\begin{tikzpicture}[scale=1]
\node (plt) at (0,0) {\includegraphics[width=0.8\textwidth]{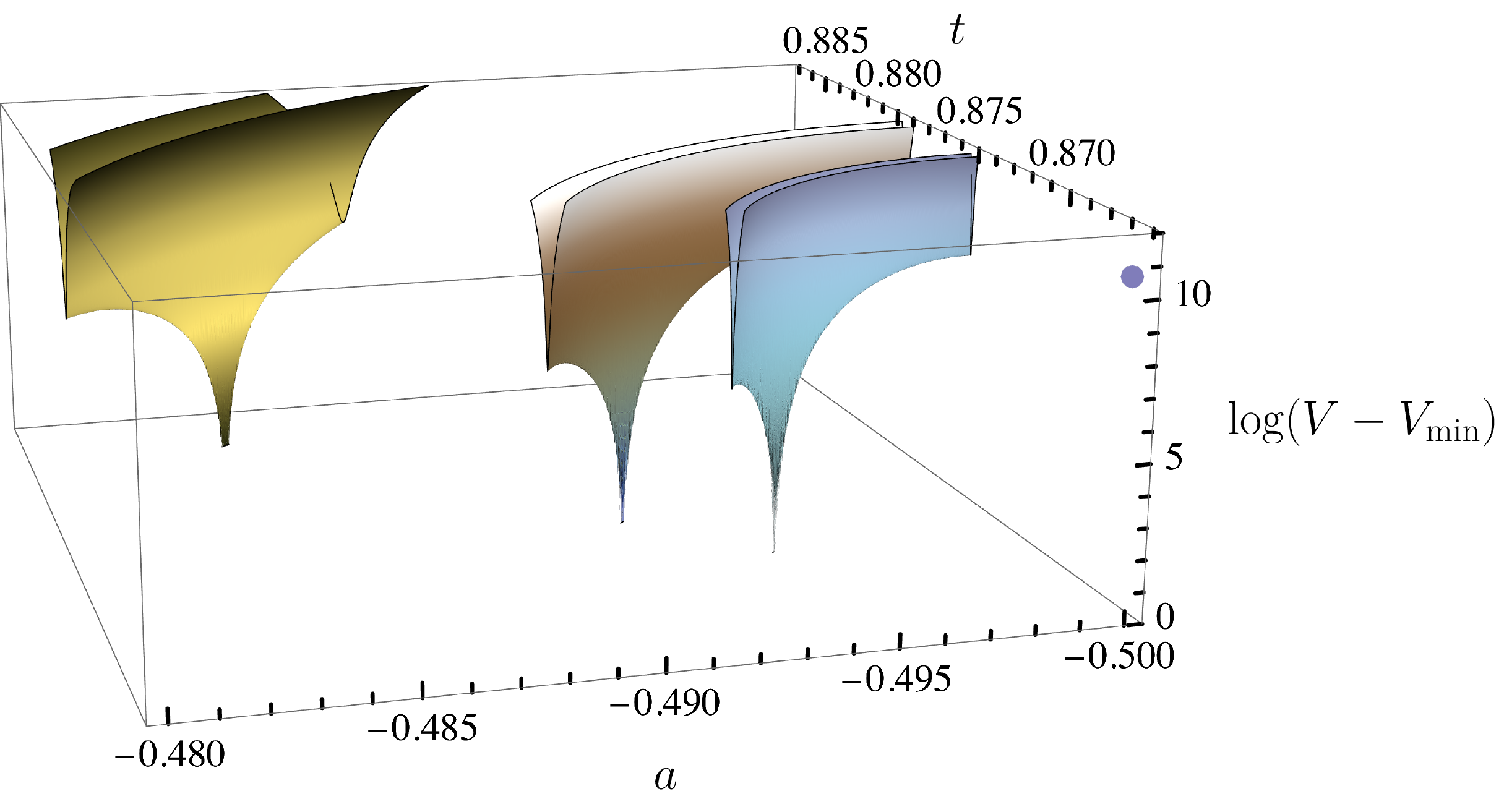}};
%\node[left= of plt,node distance=0cm,yshift=-0.5cm,xshift=1.cm] {$\log(V-V_{\text{min}})$};
 %\node[below= of plt,node distance=0cm,yshift=1.2cm,xshift=1.5cm] {$a$};
% \node[right= of plt,node distance=0cm,yshift=3cm,xshift=-2cm] {$t$};
 \node[node distance=0cm,yshift=2.8cm,xshift=-4.3cm] {\small{$m=1$}};
\node[node distance=0cm,yshift=2.4cm,xshift=-0.55cm] {\small{$m=2$}};
\node[node distance=0cm,yshift=1.8cm,xshift=1.2cm,white] {\small{$m=3$}};
   \end{tikzpicture}
\caption{Minima inside the fundamental domain close to the fixed point $T=\rho$ at varying $m$ and fixing $n=0$ in the function $H(T)$. The higher $m$, the closer the minimum is to the point $\rho$, which in this setup is a maximum and is represented by the blue dot. Each plot is appropriately rescaled.}
    \label{fig:fund_dom_minima}
\end{figure}

Other than the above exceptional cases, we also found many minima that lie on the boundary points of $\widetilde{\mathcal{F}}$ and $\text{Re}(T)=0$. As anticipated by Theorem 1 and Corollary 1.2, all of the minima are AdS. However, in general, vacua away from the fixed points have non-vanishing values for $F_T$. Thus the criterion for positive vacuum energy in Class B-3 vacua is more relaxed as we only require:
 \be\label{eq:boundouterrim}
    A(S,\Sb)>3-\frac{\widehat{V}(T,\bar{T})}{\lvert H(T)\rvert^2}\fstop
\ee
A simple realization of the aforementioned modular landscape utilizing this condition can be seen in \cref{fig:outerrimpot}, where we give an example of an AdS minimum which is uplifted to dS once the requirement in \cref{eq:boundouterrim} is met. Note that this does not destabilize the fixed point $T=\rho$, which remains a Minkowski minimum.  

To truly confirm that Class B-3 dS vacua are possible, one must analyze the $S$ and $T$ sectors together. We have done this for several examples utilizing the framework on Shenker-like effects in the following section. We did not perform a general scan, but in principle it is indeed possible to stabilize Class B-3 vacua with positive vacuum energy utilizing the Shenker-like terms.

\begin{figure}[t]
    \centering
\includegraphics[width=0.7\textwidth]{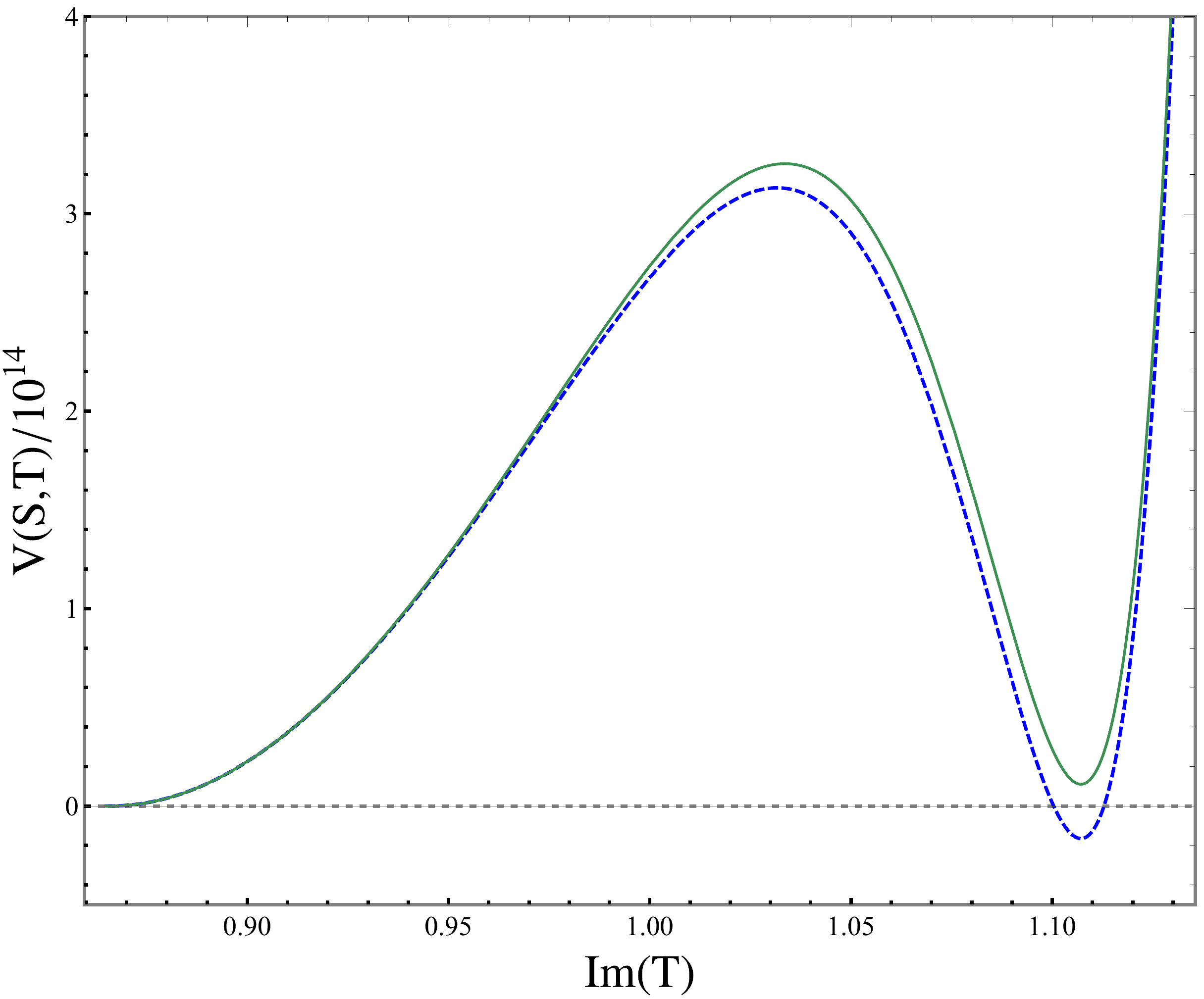}
 \caption{Potential for the parameters $\{\beta=10^{-3},\, m=2,\, n=2\}$. The AdS global minimum is located at $T=-0.5+1.107i$ while the fixed point $T=\rho$ is a Minkowski local minimum as expected from \cref{sec:typeIIvacua}. When we turn the dilaton dependence on and $A(S,\Sb)$ satisfies \cref{eq:boundouterrim}, the minimum in the outer rim becomes a dS local minimum, while $T=\rho$ stays a Minkowski minimum.}
  \label{fig:outerrimpot}
\end{figure}

%========================================================
\section{Non-Perturbative Stringy Effects \& de Sitter Vacua}
\label{sec:npeffects}
%========================================================

In the previous section, we described conditions on the function $A(S,\sbar)$ such that the two-modulus potential could permit dS vacua. Traditional heterotic model building techniques do not furnish appropriate mechanisms -- the use of $H_3$-flux or racetracks typically fix the dilaton vev in Class A vacua such that $\langle F_S\rangle=0$ and so $A(S,\sbar)$ identically vanishes. Therefore, to either extend our no-go theorem or find dS vacua, we are naturally led to effects beyond gaugino condensation.

\subsection{Class B de Sitter No-Go Theorem}
In the above, we have seen that Class B extrema, with the assumption of a stabilized dilaton, can in principle have positive energy while stabilizing the K\"{a}hler modulus. However, this class of solutions is not a panacea  -- a broad subclass of Class B extrema are unstable in the dilaton subsector. A version of this statement can be found in~\cite{deCarlos:1992kox}, which we encapsulate and extend in the following no-go theorem:

\vspace{0.2cm}

%\textbf{Theorem 2:} 
\begin{theorem}
 At a point $(T_0,S_0)$, the scalar potential in~\cref{eq:pot1} with\\            \indent\hspace{1.7cm} $k(S,\bar{S}) = -\ln(S+\bar{S})$ can not simultaneously satisfy:
\begin{enumerate}[label=(\roman*).,leftmargin=4.8\parindent]
    \item $V(T_0,S_0) > 0$
    \item $\partial_SV(T_0,S_0)= 0 \quad \&  \quad \partial_TV(T_0,S_0) = 0 $
   % \item $(\Omega_S+k_S\Omega)\rvert_{S=S_0}\neq0$
    \item $\widetilde{F}_T(T_0)=0$
    \item Eigenvalues of the Hessian of $V(T,S)$ at $(T_0,S_0)$ are all $\ge 0$.\vspace{0.2cm}
\end{enumerate}
\end{theorem}

\begin{proof}
The proof proceeds in a similar fashion as the proof of Theorem 1. We assume (i)-(iv) are true and find a contradiction. To reconcile (i) and (iii), we must be in a Class B extrema such that~\cref{eq:classBcond} and its complex conjugate are satisfied. For (i) to be true, we also require that $H(T_0)\equiv H_0\neq 0$ and we can once again introduce a parameter $\Lambda^4>0$ such that
\be
    V(T_0,S_0) = e^{k_0}Z_0 \lvert H_0\rvert^2\Lambda^4\fstop
\ee
Once again the subscript 0 denotes evaluation at $(T,S)=(T_0,S_0)$. This can be solved with
\be
    (\Omega_{S})_0 = - (k_{S})_0 \Omega_0 \pm \sqrt{(k_{S\bar{S}})_0}\bigg\{
                        \Lambda^2  \pm i \bigg( \frac{3\lvert\Omega_0\rvert^2\lvert H_0\rvert^2 - \lvert\Omega_0\rvert^2\widehat{V}_0}{\lvert H_0\rvert^2}\bigg)^{\frac{1}{2}}\bigg\}\coma  
\ee
where any of the four sign combinations are valid. 
Note that $k(S,\bar{S})$ is given its tree-level expression in the assumptions of the theorem but we keep it general for now. Similarly, $\widehat{V}_0$ vanishes by (iv) and we will set it to zero below. For the second condition in (ii), note that
\be\begin{split}
    \partial_T V(T,S) =
    e^kZ\lvert \Omega\rvert^2\bigg\{
        \bar{H}(A-3)\widetilde{F}_T + \partial_T\widehat{V} - \frac{3i}{2\pi}\widehat{G}_2\widehat{V}\bigg\}\fstop
    \end{split}
\ee
Every term in the above derivative is proportional to $\widetilde{F}_T$ or its complex conjugate and so (iii) ensures the entire expression vanishes. By a similar logic, all mixed $S$ and $T$ derivatives vanish:
\be
    \partial_S^k\partial_{\bar{S}}^l \partial_T^{} V(T_0,S_0) = 0 \quad\quad \forall\,\, k,\,l\in\mathbb{N}_+\fstop
\ee
Hence the Hessian is block diagonal. We now examine the eigenvalues of the dilaton subsector. The components of the Hessian in terms of the fields $S = s+ib$ are
\begin{align}
    \partial_s^2 V &= 2\partial_S\partial_{\bar{S}}V + 2 \text{Re}(\partial_S^2V)\coma\\
    \partial_b^2 V &= 2\partial_S\partial_{\bar{S}}V - 2 \text{Re}(\partial_S^2V)\coma\\
    \partial_s\partial_bV &= -2 \text{Im}(\partial_S^2 V)\fstop
\end{align}
First, we see that the expressions for $\Omega_S(S_0)$ and $\Omega_{SS}(S_0)$ imply that
\begin{align}
    \label{eq:genSSbar}
    \partial_S\partial_{\bar{S}}V(T_0,S_0) &= \frac{e^{k_0}Z_0\lvert\Omega_0\rvert^2}{(k_{S\bar{S}})^2_0} \bigg[\left(2\lvert H_0\vert^2\lvert \Omega_0\vert^2+\Lambda^4 \widehat{V}_0\right)(k_{S\bar{S}})^3_0\\&\,\qquad + \left((3-\widehat{V}_0)\lvert\Omega_0\rvert^2+\Lambda^4\lvert H_0\rvert^2\right)\left(k_{S\bar{S}\bar{S}}k_{SS\bar{S}} - k_{S\bar{S}}k_{SS\bar{S}\bar{S}}\right)_0\bigg]\nonumber\\
    &\rightarrow -\frac{2Z_0\lvert H_0\rvert^2(2\lvert \Omega_0\rvert^2+\Lambda^4)}{(S_0+\bar{S}_0)^3}\fstop
\end{align}
In going to the second line we set $k(S,\bar{S}) = -\ln(S+\bar{S})$ and $\widehat{V}(T_0,\bar{T}_0)=0$. 
Then by identical logic to Theorem 1, we see that at least one eigenvalue of the Hessian is negative. Logically speaking, we must supplement the assumptions of the theorem with $\text{Re}(S_0)>0$, but this is a physical requirement for a sensible coupling constant. Thus (iv) is incompatible with (i)-(iii) and the theorem is demonstrated.
\end{proof}

\noindent In~\cref{sec:anomaly} we prove a similar statement to the above for orbifolds with non-trivial mixing in the K\"{a}hler potential displayed in~\cref{eq:modinvKahpot}. For the moment, we find an immediate corollary:

\vspace{0.2cm}

\begin{corollary}
Extrema of the two-modulus model with $k(S,\bar{S})=-\ln(S+\bar{S})$\\\indent \hspace{2.15cm} and~\cref{eq:superpot1} can never be dS vacua at the fixed points of $\pslz$.\vspace{0.2cm}
\end{corollary}

\noindent This follows from Theorems 1 and 2 and the fact that $\widetilde{F}_T$ vanishes at the $\pslz$ fixed points for all cases except for $(m,n) = (1,n)$ at $T=i$ and  $(m,n) = (m,1)$ at $T=\rho$ -- see~\cref{app:modularforms}. However, as discussed in~\cref{sec:vacua}, these cases can only ever yield unstable dS and the corollary is verified. 

Thus we have an analytic argument for further results of~\cite{Gonzalo:2018guu}. They found by numerical analysis that Class B vacua at the fixed points were never dS vacua in cases where $\Omega(S)$ was a sum of exponentials, which is reminiscent of a racetrack superpotential. We see that this result and more is captured by Corollary 2.1 -- for a tree-level dilaton K\"{a}hler potential, no racetrack or other non-perturbative effect captured by $\Omega(S)$ in the superpotential can result in dS vacua at $T=i,\,\rho$. The same applies for $H_3$-flux appearing as an additive constant in $\Omega(S)$. This result is a natural extension of~\cite{Quigley:2015jia} to orbifold models with threshold corrections and worldsheet instantons. 

Theorem 2 illustrates that not any Class B extremum can be a dS vacua -- in particular, constructing arbitrarily complicated superpotentials via $\Omega(S)$ is insufficient. However,~\cref{eq:genSSbar} indicates how to go beyond Theorem 2 at the fixed points -- one must go beyond the tree-level K\"{a}hler potential for the dilaton.\footnote{One could also analyze~\cref{eq:gensuperpot}, which is not directly covered by Theorems 1 \& 2.} This can be achieved by the Shenker-like terms briefly described in the introduction. We now turn to reviewing and modeling these effects.

\subsection{Stringy Effects in Heterotic Theories}
As emphasized above, gaugino condensation is not an inherently stringy phenomenon -- rather, it is a purely quantum field-theoretical effect. This is evident from the form of the non-perturbative superpotential in~\cref{eq:superpotp5}, which contributes terms to the Lagrangian of the form $\delta\mathcal{L}\sim e^{-1/g_s^2}$. Truly stringy non-perturbative effects would scale as $\delta\mathcal{L}\sim e^{-1/g_s}$, which can be found via D-branes in theories with open strings. Contributions from such stringy effects would be stronger than gaugino condensation at weak coupling. 

Considering such effects in the context of heterotic models at first seems surprising, given that these string theories lack D-branes. Nonetheless, heterotic models can and do have these stringy non-perturbative effects. 
The original argument for the existence of such effects by Shenker~\cite{Shenker:1990} is based on the scaling properties of matrix theory amplitudes and applies to all closed string theories.

Shenker made a generic existence argument, but one can go a bit further by leveraging string dualities~\cite{Silverstein:1996xp,Antoniadis:1997nz}.  Indeed, Silverstein~\cite{Silverstein:1996xp} extended these general arguments by outlining how such non-perturbative stringy effects could be seen from the dualities of heterotic theories with type I and type IIA string theories. Let us first consider the duality with type I. The $10d$ type I/heterotic map is
\be\renewcommand{\arraystretch}{1.3}
\begin{array}{rcl}
   G^H_{MN} &= &g_s^H G_{MN}^I  \\
   g_s^H  & =&(g_s^I)^{-1}\coma
\end{array}
\ee
where $G_{MN}^i$ is the $10d$ metric of heterotic or type I theory. Thus, type I worldsheet instantons wrapping a 2-cycle with area $A^I$ that give contributions $\delta\mathcal{L}\sim e^{-A_I}$ get mapped to $\delta\mathcal{L}\sim e^{-A_H/ g_s^H}$. This is precisely the type of Shenker-like effect described above. 

On the type IIA side, one can consider a scenario where the heterotic dilaton is mapped to a K\"{a}hler modulus on the type IIA side: $S^H \leftrightarrow T^{IIA}$. Such a duality is known to hold also in $4d$ in the large volume/weak coupling limit~\cite{Kachru:1995wm,Vafa:1995gm}. In this case, Silverstein considers ``worldline" instanton contributions $\delta\mc{L}\sim \sum_m e^{-m r}$ arising from strings wrapping a non-contractible loop of radius $r$ in the compactification manifold. This implies that the compactification manifold must have a non-trivial fundamental group. Applying the duality, we get in the heterotic side an effect which again falls like $\delta\mc{L}\sim \sum_m e^{-m/g_s^H}$.

The above was made much more tangible via the explicit calculations in $10d$ and $9d$ heterotic string theories by Green and Rudra in~\cite{Green:2016tfs}. The authors calculated one-loop diagrams in the Ho\v{r}ava-Witten background of $11d$ supergravity and find non-trivial $R^4$ terms which can be cast as contributions in the heterotic action proportional to $e^{-1/g_s}$. For example, a one-loop $4$ graviton bulk amplitude contributes a term to the $10d$ Spin$(32)/\mathbb{Z}_2$ heterotic effective action of the form
\begin{align}
    S^{HO} \supset \frac{g_{HO}^{-1/2}}{2^8(2\pi)^74!\,l_H^2}\int_{\mathcal{M}_{10}}\sqrt{-G}\,t_8t_8 R^4\mathbf{G}_{3/2}(ig^{-1}_{HO})\coma
\label{eq:HOShenker}
\end{align}
where $t_8$ is a rank-eight tensor and $l_H$ is the heterotic string length. The dependence on the string coupling $g_{HO}$ is encoded in the real-analytic Eisenstein series $\mathbf{G}_s(\tau)$. These are non-holomorphic modular forms with weight $(0,0)$ (i.e. non-holomorphic modular functions) -- see~\cref{app:modularforms}. For $s=3/2$, the series has a small-$g_s$ expansion
 \begin{align}
     \mathbf{G}_{3/2}(ig_s^{-1}) \!&= \zeta(3) g_{s}^{-3/2} +2\zeta(2)g_{s}^{1/2}+\sum_{n\in\mathbb{Z}^+}4\pi\sigma_{-1}(\lvert n\rvert)e^{-\frac{2\pi\lvert n\rvert}{g_{s}}}(1+\mathcal{O}(g_{s}))\\
                        \!&= \zeta(3) g_{s}^{-3/2} +2\zeta(2)g_{s}^{1/2}+e^{-\frac{2\pi}{g_s}}\left(4\pi + \frac{3}{4} g_s +\mathcal{O}(g_s^2)\right)+\mathcal{O}(e^{-\frac{4\pi}{g_s}}) \, ,
\label{eq:realeis}
 \end{align}
$\sigma_s(n)$ being the divisor sum, as defined in~\cref{app:modularforms}. Thus we see that there is an infinite set of non-perturbative instanton contributions of the form predicted by Shenker~\cite{Shenker:1990}. This contribution to the action is not unlike the $R^4$ term in type IIB string theory, where the $g_s^{-1}$ dependence is augmented to a dependence on the axio-dilaton. Indeed one can think of the result in~\cref{eq:HOShenker} as being related to the $9d$ $R^4$ term in IIB via orientifold projection to type I, S-dualizing to HO, and then taking the $10d$ limit. The authors find a similar structure in the gauge sector via a non-perturbative contribution to the $F^4$ correction to the effective action. Interestingly,~\cite{Green:2016tfs} find these instanton effects for both heterotic string theories in $9d$, but the calculated Shenker-like effects vanish in the $10d$ limit of the $E_8\times E_8$ heterotic string. This makes a tantalizing connection with the discussion of stable open heterotic strings by Polchinski~\cite{Polchinski:2005bg}.

We close this discussion with a note on the importance of duality in the above results for the heterotic Shenker-like terms in $9d$ and $10d$. As argued in~\cite{Green:2016tfs}, S-duality between type I and HO cannot operate term-by-term in~\cref{eq:HOShenker} -- if we write the $R^4$ term of the $9d$ HO and type I theories as $r_i\ell_i^{-1}g_i^{-\frac{1}{2}}U(ig_i^{-1})t_8t_8R^4$ with $r_i$ the radius of the additional $S^1$, $\ell_i$ the string length, $g_i$ the string coupling, and $i= I, HO$, then S-duality demands
\begin{equation}
    U(ig_{HO}^{-1}) =  U(ig_{I}^{-1})
\end{equation} 
since $r_{HO}\ell_{HO}^{-1}g_{HO}^{-\frac{1}{2}}=r_I\ell_I^{-1}g_I^{-\frac{1}{2}}$. This property is not satisfied by the perturbative part of the $R^4$ coefficient, indicating that one must include additional terms. On the other hand, it is satisfied by the real-analytic Eisenstein series -- it follows from their invariance under S-transformations of $\pslz$. Thus one could argue that heterotic Shenker-like effects play an important role in the duality web connecting string theories.
 
 This exhausts the discussion of heterotic stringy non-perturbative effects in the literature, as far as the current authors are aware. Clearly, these effects are poorly understood and deserve thorough investigation in their own right. However, in the following section we will attempt to understand the impact of these effects on heterotic vacua in the effective $4d$ theories considered above and leave an explicit full-fledged derivation for future work.  
Along a similar vein, these Shenker-like effects have been previously applied to the program of heterotic particle phenomenology in the form of ``K\"ahler Stabilized Models"~\cite{Casas:1996zi,Binetruy:1996nx,Binetruy:1997vr,Barreiro:1997rp,Kaufman:2013pya} -- see~\cite{Gaillard:2007jr} for a review.

\subsection{Stringy de Sitter Vacua \& the Linear Multiplet Formalism}
As argued above, heterotic string theories contain non-perturbative corrections to their effective actions with strength $\mathcal{O}(e^{-1/g_s})$. In the $4d$ models we have introduced, these corrections will appear in the K\"{a}hler potential of the low energy effective theory. To study the impact of these non-perturbative effects on the two-modulus model, we must settle for a parametrization in the absence of explicit calculations. For concrete calculations, we will assume that the Shenker-like terms are functions of solely the dilaton and therefore they have the characteristic strength of $\mathcal{O}(e^{-1/g_4})$. We will return to the question of K\"{a}hler modulus dependence at the end of this section. Furthermore, we follow the conventions of~\cite{Gaillard:2007jr} and use the linear multiplet formalism for the dilaton~\cite{Ferrara:1974ac,Siegel:1979ai}.  This has been advocated as the superior convention for the dilaton~\cite{Derendinger:1994gx,Gates:1988kp,Siegel:1988qu} since it naturally describes corrections to the gauge coupling, Green-Schwarz anomaly cancellation conditions~\cite{Butter:2014eva,Gaillard:2017ypo,Gaillard:2019ofq}, effective descriptions of gaugino condensation, and higher-genus corrections.\footnote{However, care should be taken with the linear multiplet after the dilaton gains a mass~\cite{Burgess:1995aa}. We thank Fernando Quevedo for raising this point.}
The component fields of the linear multiplet superfield $L$ are a real scalar $\ell$ (the dilaton), a Majorana fermion $\psi$ (the dilatino), and an antisymmetric tensor (the 2-form $B_2$). This decomposition also motivates the linear multiplet formalism as the natural choice since its degrees of freedom are precisely those of the string compactification. 

In the $U(1)_K$ superspace formalism~\cite{Binetruy:2000zx}, the kinetic contribution to the effective supergravity action is determined by the superspace integral
\begin{align}
    \mathcal{L}_{KE} &= \int d^4\theta E(-2+f(L))\coma
\end{align}
where $E$ is the supervielbein determinant. The function $f(L)$ parametrizes the Shenker-like terms described above. It is related to the dilaton ``K\"{a}hler potential"\footnote{This is not a true K\"{a}hler potential since the dilaton is no longer described by a chiral superfield and is more accurately described as a ``kinetic potential".}
\begin{align}
        k(L) &= \ln(L) +g(L)\coma
\end{align}
via the differential equation
\begin{align}
     L\,\frac{df}{dL}&=- L\,\frac{dg}{dL} +f\fstop
\label{eq:dildiffeq}
\end{align}
In the presence of the non-perturbative terms, the relation between $g_4$ and the dilaton vacuum expectation value is altered:
\be
            \frac{g_4^2}{2} = \bigg\langle\frac{\ell}{1+f(\ell)}\bigg\rangle\fstop 
\ee

In the linear multiplet formalism, gaugino condensation does not induce a superpotential for the dilaton. Instead, the scalar potential of the K\"{a}hler modulus and dilaton model can be calculated from the corrections to the gauge kinetic function, as written in~\cref{eq:gaugekin}. See~\cite{Gaillard:2007jr} for details. If we consider a single condensing gauge group with beta function coefficient $b_a$ (corresponding to $\Omega(S) = e^{-S/b_a}$ in the chiral formalism of~\cref{sec:heterotic}) the scalar potential of~\cref{eq:pot1} in the linear multiplet formalism becomes
\be
    V(\ell,T,\bar{T}) =  Z(T,\bar{T})\bigg[\! \bigg(\frac{(1+b_a\ell)^2(1+\ell g'(\ell))}{b_a^2\ell^2}-3\!\bigg)\lvert H(T)\rvert^2+\widehat{V}(T,\bar{T})\bigg]\ell e^{g(\ell)-(f(\ell)+1)/b_a\ell}\fstop
\label{eq:lindilpot}    
\ee
For a single gaugino condensate, one can bypass superspace calculations and derive this potential simply by taking~\cref{eq:pot1} and implementing the tree-level relation between chiral and linear multiplet formalisms:
\be
        \frac{\ell}{1+f(\ell)} = \frac{1}{S+\bar{S}}\fstop
\label{eq:chilinmap}
\ee
To proceed, we must specify the functions $f(\ell)$ and $g(\ell)$ that parameterize the non-perturbative stringy terms. Due to the differential equation in~\eqref{eq:dildiffeq}, we need only specify one of the functions and a boundary condition. As in~\cite{Gaillard:2007jr}, we will specify $f(\ell)$ and fix $g(\ell)$ by the requirement that $g(0)=0$. This last condition is simply demanding that the non-perturbative effects vanish as the string coupling vanishes and only gets rid of a constant in the solution for $g(\ell)$. Note that one could also specify $g(\ell)$ and then derive $f(\ell)$. %The choice is somewhat arbitrary, but $f(\ell)$ is the function that appear directly in the action
We take $f(\ell)$ to have the form~\cite{Gaillard:2007jr}
\be\label{eq:npcorrlinear}
    f(\ell) = \sum_{n=0}A_n \ell^{q_n}e^{-B/\sqrt{\ell}}\coma
\ee
with constants $A_n$, $q_n$, and $B$. We will always consider $B>0$ with the expectation that non-perturbative effects should vanish at exponentially at large field values~\cite{Witten:1996bn}. The polynomial in $\ell$ is mirroring the structure of the real-analytic Eisenstein series expansion in~\cref{eq:realeis}. The above is essentially a parametrization for the first term in the infinite series of~\cref{eq:realeis}. In principle, we should include these higher order instantons. In practice, if our vacuum yields a small value for the first instanton, we will neglect higher order terms. We will return to this point below. 

With this parametrization of $f(\ell)$, $g(l)$ is
\begin{align}
    % f(\ell) &=\sum_{n=0}A_n \ell^{-q_n}e^{-B/\sqrt{\ell}}\coma\\
    g(\ell) &= \sum_{n=0}A_n  B^{2q_n}\bigg\{ 2(1-q_n)\Gamma(-2q_n,B/\sqrt{\ell})-\Gamma(1-2q_n,B/\sqrt{\ell})           \bigg\}\coma
\label{eq:npcorrlinearg}
\end{align}
where $\Gamma(a,x)$ is the upper incomplete gamma-function
\be
    \Gamma(s,x) = \int_{x}^\infty y^{s-1}e^{-y}dy \fstop
\ee
We now consider specific examples of~\cref{eq:lindilpot} using~\cref{eq:npcorrlinear} and~\cref{eq:npcorrlinearg}. To simplify the discussion, we will consider vacua at the fixed point $T=\rho$. So long as $n=0$ in the parametrization of $H(T)$,~\cref{eq:Hpara}, any dS minimum in the dilaton sector is immediately a minimum of the $T$ sector. Furthermore, $\widehat{V}(T,\bar{T})$ in~\cref{eq:lindilpot} vanishes. We can also set $m=0$ without loss of generality since non-zero $m$ simply modulates the value of the vacuum energy via an overall factor in~\cref{eq:lindilpot}.

\paragraph{Example 1: Trivial polynomial}
First, we consider the case of a trivial polynomial setting $q_n=0$ $\forall$ $n$:
\begin{align}
    f_1(\ell) &= A_0e^{-B/\sqrt{\ell}}\coma\\ 
    g_1(\ell) &= A_0\bigg\{ 2\Gamma(0,B/\sqrt{\ell}) - e^{-B/\sqrt{\ell}}        \bigg\}\fstop
\end{align}
If we take $B= \pi$ and $A_0=26$, we find the potential on the left-hand side of~\cref{fig:dSexamples}. At the dS vacuum, we find
\begin{align}
    g_4 &\simeq 0.99\coma\nonumber\\
    \langle f(\ell)\rangle  &\simeq 1.92\coma\\
    \langle e^{-B /\sqrt{\ell}}\rangle &\simeq 7.4\times 10^{-2}\nonumber\fstop
\end{align}
\begin{figure}
\centering
\includegraphics[width=\textwidth,keepaspectratio]{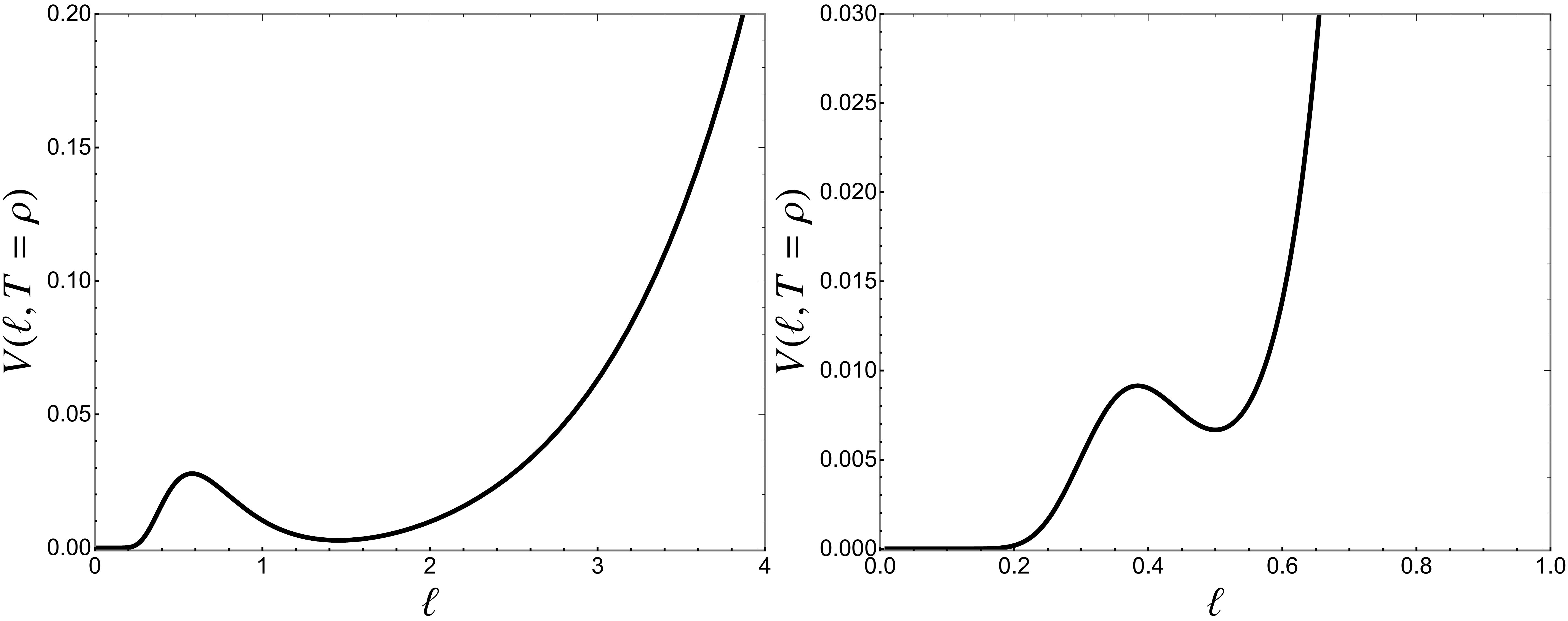}
\caption{Plots of the scalar potential evaluated at $T=\rho$ for Example 1 (left) and Example 2 (right).}
   \label{fig:dSexamples}
\end{figure}
\paragraph{Examples 2 \& 3: Non-trivial polynomials}
Let us also consider two examples with non-trivial polynomials. We will keep the first two terms in \cref{eq:npcorrlinear}:
\be
f_{2,3}(\ell) = (A_0+A_1\ell^{q})e^{-B/\sqrt{\ell}}\fstop
\ee
If we take $q=-0.5$, $B=2\pi$, $A_0 = -7372$, and $A_1=7254$, we find the potential on the right in~\cref{fig:dSexamples}
At this minimum, the parameters take the values 
\begin{align}
    g_4 &\simeq 0.84\coma\nonumber\\
    \langle f(\ell)\rangle  &\simeq 0.39\coma\\
    \langle e^{-B/\sqrt{\ell}}\rangle &\simeq 1.3\times 10^{-4}\nonumber\fstop
\end{align}
Instead, if we take $q=1/2$, $B=0.6\pi$, $A_0 = 10$, $A_1 = 9$
\begin{align}
    g_4 &\simeq 0.70\coma\nonumber\\
    \langle f(\ell)\rangle  &\simeq 2.1\coma\\
    \langle e^{-B/\sqrt{\ell}}\rangle &\simeq 0.11 \nonumber\fstop
\end{align}
This final example is motivated by the Eisenstein expression from~\cite{Green:2016tfs}, where the polynomial coefficients are $\mathcal{O}(1-10)$. In both examples 2 and 3, the string coupling is stabilized close to the oft quoted phenomenological target of $g_4 = 2^{-1/2}\simeq 0.707$.

Let us make several comments on these vacua. First, note that the relationship between the linear and chiral multiplet formalisms in~\cref{eq:chilinmap} implies that the typical dilaton runaway $\text{Re}(S)\rightarrow 0$ to weak coupling maps to $\ell\rightarrow 0$. Examination of the potentials in~\cref{fig:dSexamples} shows that vacuum tunneling is possible via bubble nucleation to Minkowski space with $g_s\rightarrow 0$.  This generic behavior of the vacua is rooted in the vanishing of the Shenker-like effects at weak coupling. Thus the dS vacua generated by Shenker terms are \textit{metastable}. This is in accord with the general notion that if dS vacua exist in theories of quantum gravity, they must be at best metastable~\cite{Goheer:2002vf} and perhaps correspond to resonances~\cite{Maltz:2016iaw,Brahma:2020tak,Bernardo:2021rul} or mixed states in the AdS/CFT correspondence~\cite{Kaloper:1999sm,Hawking:2000da,Freivogel:2005qh,Maldacena:2010un}.

Second, an important concern regarding the above vacua is control. Perturbative heterotic $4d$ string vacua which generate weakly-coupled non-abelian gauge groups with values of the gauge couplings roughly of that of simple GUTs $\alpha_{\rm GUT}\simeq 1/25$ cannot exist at large volume ${\cal V}$ of the compactification manifold. This statement follows from the tree-level relation between the value of the tree-level heterotic $4d$ gauge coupling $\alpha_{4d}$, the $4d$ heterotic dilaton $\text{Re}(S)$, the string coupling $g_s$, and the volume ${\cal V}$ of the extra dimensions
\begin{equation}
\alpha_{4d}^{-1}=4\pi \text{Re}(S) =\frac{\cal V}{g_s^2}\fstop
\end{equation}
If we now require $\alpha_{4d}^{-1}\simeq \alpha_{\rm GUT}^{-1}\simeq 25$ and $g_s<1$ to keep the heterotic string perturbative, this implies small compactification volume ${\cal V}< 25$. Hence, the control mechanism of a large-volume expansion traditionally employed in the context of type IIB string vacua is not operational here.

Instead, four-dimensional heterotic strings, on orbifolds in particular, allow for the computation of perturbative and non-perturbative quantum corrections to the space-time effective action directly using worldsheet conformal field theory techniques. These corrections in turn are highly constrained by holomorphy and modular invariance, leading to a level of control (in some cases to infinite order of a particular series of quantum corrections) that can match and sometimes even surpass the level of large-volume expansion control in type IIB string vacua.

While the dilaton has been stabilized such that $g_s$ is in a perturbative regime, our parametrization of the Shenker-like effects in~\cref{eq:npcorrlinear} ostensibly includes only the lowest order term in an infinite sum of non-perturbative terms. One would imagine that the full set of Shenker-like corrections has the form
\be
    f(\ell) = \sum_{\alpha=0}^\infty P_\alpha(\ell)e^{-B_\alpha/\sqrt{\ell}}\coma
\label{eq:dinstsum}
\ee
for some polynomials $P_\alpha(\ell)$ and constants $B_\alpha$. Indeed this is what the M-theory calculations of~\cite{Green:2016tfs} and duality arguments in~\cite{Silverstein:1996xp} point towards. Viewed from this perspective, consistency of the vacua requires that the high-order terms must be subdominant to the lowest-order one we have included. There are two ways this could be achieved, depending on the nature of the sum in~\cref{eq:dinstsum}. If the sum behaves similar to an instanton sum from Gopakumar-Vafa invariants, then the coefficients of the terms in~\cref{eq:dinstsum} could grow rapidly as $\alpha$ increases. Control then demands that the exponentials decay rapidly with $\alpha$. Such a scenario would correspond to Example 2, where the ratio of the lowest-order exponential to the nominal value of the next exponential is $\langle e^{-2B/\sqrt{\ell}}/e^{-B/\sqrt{\ell}}\rangle = \langle e^{-B/\sqrt{\ell}}\rangle =  \mathcal{O}(10^{-4})$. On the other hand, it may be that the coefficients in~\cref{eq:dinstsum} do not grow rapidly, or at all, with $\alpha$. This would occur if the sum in~\cref{eq:dinstsum} descends from a function such as the real-analytic Eisenstein series of~\cite{Green:2016tfs}. For such a scenario, control could be achieved with even modestly small exponentials, such as Example 3 above where the instanton ratio is only $\mathcal{O}(10^{-1})$. In either case, control hinges on the precise nature of the sum of Shenker-like terms. We are not in a position to argue generalities about this sum and leave a detailed discussion to future work.

\paragraph{} Finally, we have downplayed the role of the K\"{a}hler modulus in this section and simply fixed $T=\rho$ at the very start. This approach is valid solely because of the remarkable number-theoretic properties of the scalar potential. Modular symmetry provides a simple means to identify critical points and vacuum criteria via the block-diagonal structure of the Hessian at $\pslz$ fixed points. As shown in~\cref{sec:typeIIvacua}, $T$ is stabilized at $T=\rho$ if $A(S,\bar{S})>2$  -- a condition that can, in principle, be satisfied by the Shenker-like terms.

On the other hand, we have parametrized the Shenker-like terms solely as a function of the dilaton. However, it is conceivable that these stringy effects depend on the K\"{a}hler modulus. While this would complicate the above calculations, there is a conceptual issue as well -- an arbitrary dependence on the K\"{a}hler modulus would break the modular symmetry of $T$ and hence T-duality, which would raise concerns over the consistency of the string models. Therefore, we posit that the Shenker-like terms must uphold T-duality and consider three scenarios for the sum of non-perturbative effects: 
\begin{itemize}
    \item[] \textbf{Case 1:} the stringy effects do not depend on the K\"{a}hler modulus at all
    \item[] \textbf{Case 2:} the stringy effects depend  on $T$ only through a modular invariant function
    \item[] \textbf{Case 3:} a single instanton term breaks T-duality, but the sum of instantons is a function with well-defined modular properties.
\end{itemize} 
We argue that all the scenarios above may be realized in varied heterotic compactifications. For Case 1, the schematic calculations of~\cite{Silverstein:1996xp} indicate that the type IIA-heterotic map results in Shenker-like terms that do not depend on the K\"{a}hler modulus, although they may depend on other moduli that we have omitted in the considerations of this paper. 

The second case above occurs in $4d$, $\mathcal{N}=2$ type IIB-heterotic dualities~\cite{Candelas:1993dm,Kachru:1995wm,Kachru:1995fv}. One example is type IIB on the mirror dual of $\mathbb{P}^4_{11226}[12]$, which is dual to heterotic on $K3\times T^2$. One combination of the type IIB complex structure moduli, $x$, gets mapped to the heterotic  $T^2$ K\"{a}hler modulus via the j-invariant:
\be
        x = \frac{1728}{j(T)}+\cdots\fstop 
\ee
The dots are corrections that arise away from the strict heterotic weak coupling limit. Proposals for dualities involving similar type IIB-heterotic maps to the above but with hauptmoduls of congruence subgroups of $\pslz$ were studied in~\cite{Klawer:2021ltm,Alvarez-Garcia:2021mzv}.

The third possibility would be realized if the instanton sum was itself the expansion of some modular-invariant function, such as the real-analytic Eisenstein series in~\cref{eq:RAEis} and similar to the results of~\cite{Green:2016tfs}. A single term in the sum breaks modular invariance, but the entire sum is invariant. 

Let us reconsider our vacua in light of these scenarios for the sum of Shenker-like terms. First, since T-duality is preserved, the scalar potential remains a modular function, and so the fixed points of $\pslz$ remain extrema of the scalar potential.
Furthermore, $K_T$ retains its tree-level value at the fixed points since the $T$ derivative of the Shenker-like term sum must vanish at the fixed points. This implies that in almost all cases\footnote{Exceptions occur when $(m,n)=(1,n)$ or $(m,1)$.} $F_T$ still vanishes at the fixed points. Thus, the condition $A(S,\bar{S})>3$ for a dS vacuum remains unchanged, except that $A(S,\bar{S})$ is now $T$ dependent and should be evaluated at the fixed point. Statements on the Hessian are much more difficult without a specific form of the Shenker-like terms. Modular symmetry implies that mixed entries in the Hessian vanish and we retain a block-diagonal structure. Beyond this, conditions for a minimum are dependent on the precise form of the Shenker-like terms.

We also note that for some of the vacua inside the fundamental domain, the hierarchy between the dilaton an K\"ahler modulus mass scales can become so large that the $T$-dependence of the Shenker-like terms becomes irrelevant. This is because at the level of dilaton stabilization, $T$ becomes frozen and rigidly stabilized at a much higher mass scale. In such vacua the Shenker-like terms become effectively purely dilaton-dependent. We have exhibited in Fig.~\ref{fig:fund_dom_minima} one class of vacua with this behavior. The $(m,n)=(m,0)$ with $P(j(T))=1$ vacua depicted there show stabilization of the K\"ahler modulus and its axion in `needle-thin' vacua indicating the very high mass scale of $T$ far above the mass scale where $S$ (or $\ell$, respectively) gets stabilized.

We close this section with a comment on a puzzle posed in~\cite{Silverstein:1996xp}. As discussed above, obtaining the heterotic Shenker-like terms from duality with type IIA requires a non-trivial fundamental group on the IIA compactification manifold giving rise to worldline instantons. Similarly, one could examine type I-heterotic duality with a non-trivial fundamental group. Using the map defined above, one finds that worldline instantons on the type I side get mapped to heterotic terms of strength $\delta\mathcal{L}\sim \exp(-R/\sqrt{g_s})$, with $R$ the radius of the non-contractible loop. Assuming such contributions are not identically zero, they may be stronger at weak coupling than those predicted by Shenker. We have not included these Silverstein-like effects in our considerations and leave their explicit study and impact on heterotic vacua to future work.

\section{Discussion}
\label{sec:discussion}
In this work, we have explored the possibility of extending dS no-go results to include non-perturbative stringy effects in the context of heterotic compactifications, particularly toroidal orbifolds. Focusing on the class of models where the light moduli consist of the dilaton and the overall K\"{a}hler modulus, we partially achieved this goal for extrema where the dilaton F-term vanishes via Theorem 1. The assumptions of this theorem are that the stringy effects are solely functions of the dilaton and that the superpotential has a factorized form and respects T-duality. The no-go theorem then rules out dS vacua arising from non-perturbative effects in both $T$ and $S$ in the superpotential and in $S$ in the K\"{a}hler potential when $\langle F_S\rangle =0$. As corollaries we were able to confirm previous conjectures and numerical studies arguing against dS vacua in the pure K\"{a}hler modulus theory and in the two-modulus theory when the dilaton K\"{a}hler potential has its tree-level expression. Other partial no-go theorems in the literature extend this to include perturbative quantum corrections in the heterotic moduli K\"ahler potential. 

Theorem 2 utilizes a similar argument to Theorem 1 and forbid dS vacua at the fixed points of $\pslz$ when the dilaton has only a tree-level K\"{a}hler potential, even if $\langle F_S\rangle\neq0$. This rules out dS vacua arising purely from non-perturbative effects in the superpotential in $S$ and $T$ at the fixed points, again affirming a previous conjecture. Theorem 3 in~\cref{sec:anomaly} proves a kindred statement for the case of non-vanishing Green-Schwarz coefficient $\delta_{GS}$. 

The above results point to non-perturbative corrections to the K\"ahler potential that yield $\langle F_S\rangle\neq0$ as almost the only loophole to realize heterotic de Sitter vacua. After reviewing arguments in the literature in favor of universally present non-perturbative corrections to $\mathcal{K}$ at ${\cal O}(e^{-1/g_s})$, we demonstrated that the inclusion of these effects for heterotic orbifold vacua could lead to dilaton stabilization with F-term SUSY breaking in the dilaton sector using just a single gaugino condensate contribution in both AdS and metastable dS vacua where the K\"ahler modulus is stabilized as well, as emphasized in the models of~\cite{Casas:1996zi,Binetruy:1996nx,Binetruy:1997vr,Barreiro:1997rp,Kaufman:2013pya,Gaillard:2007jr}. Furthermore, we showed that this stabilization could also lead to heterotic dS vacua. These solutions, if they truly exist vis a vis the ${\cal O}(e^{-1/g_s})$ K\"ahler corrections in a bona fide heterotic string construction, thread the needle through several partial no-go results -- the no-go theorems we present here as well as those of~\cite{Maldacena:2000mw,Gautason:2012tb,Green:2011cn,Kutasov:2015eba,Quigley:2015jia,Gonzalo:2018guu}.

These considerations leave a pressing open problem -- determination of the precise form of the dilaton-dependent ${\cal O}(e^{-1/g_s})$ corrections to the K\"{a}hler potential, whose existence in the heterotic string so far is based on duality arguments and $11d$ computations. Aside from providing a potential mechanism to stabilize the dilaton and realize heterotic de Sitter vacua, these corrections are interesting in their own right as a window into truly stringy features of heterotic theories. A primary goal of future investigations will be to illuminate the nature of heterotic Shenker-like effects and to perform calculations to complement those of~\cite{Green:2016tfs}.

Along this analysis we reviewed how modular invariance of $4d$ heterotic orbifolds is maintained by an interplay of anomaly cancellation and 1-loop threshold corrections to the gauge kinetic function of the $4d$, ${\cal N}=1$. These corrections produce the leading $1/\eta^6(T)$ dependence of the gaugino condensate superpotential required by modular invariance. Further corrections in $T$ arise beyond leading order as effectively an infinite series of worldsheet instanton corrections constrained by modular invariance to resum into a product of integer powers of the holomorphic Eisenstein functions $G_4$ and $G_6$, and an arbitrary polynomial of the modular invariant $j$-function. We demonstrate in detail that for even simple non-trivial choices of these corrections there exist non-trivial SUSY breaking AdS vacua and critical points in the fundamental domain of the K\"ahler modulus $T$ away from its boundaries. This invalidates a long-standing conjecture that such critical points were to exist only on the boundary of the fundamental domain.

It is clear that based on these results there are several immediate directions for future work. Beyond the obvious task of establishing precise calculations for the form of the ${\cal O}(e^{-1/g_s})$ K\"ahler corrections, the structure of the modular invariant $T$-dependent non-perturbative corrections in the superpotential seems to foretell the existence of a whole landscape of SUSY breaking AdS vacua inside the $T$-fundamental domain. Determining the structure of this landscape as well as their potential upliftability to dS is a further task for future work. See also~\cite{Mourad:2016xbk,Basile:2018irz,Baykara:2022cwj} for non-SUSY heterotic AdS vacua.

Another straightaway question would be to include the presence of Wilson lines on heterotic orbifolds, as these are often used in breaking the heterotic gauge symmetry down to acceptable MSSM-like particle physics sectors. Including discrete Wilson lines typically breaks $\pslz$ modular invariance down to congruence subgroups, which would affect the gaugino condensate-induced superpotential and thereby the vacuum structure of the theory. We leave a thorough study of these cases as yet another task for future work. 
On a wider scope, generalizing our results to several K\"ahler moduli and the inclusion of the complex structure moduli of orbifolds immediately calls for extending modular invariance beyond $\pslz$ to e.g. SP$(4,\mathbb{Z})$ and more general automorphic groups which promise a landscape of heterotic vacua with an even richer number-theoretical structure. 

Finally, on the more phenomenological side, the structure of the scalar potential connecting the Minkowski and dS vacua achievable by stabilizing the dilaton with ${\cal O}(e^{-1/g_s})$ K\"ahler corrections in $F_S$-SUSY-breaking vacua at the $T$-self dual points, indicates the presence of potentially slow-roll flat saddle points. As these might innately be suitable for constructing slow-roll inflation in these stabilized heterotic vacua, and string inflation setups within heterotic string vacua are mostly unknown, we leave this, too, as a well-motivated task for the future.

\section*{Acknowledgments}
We thank Rafael \'Alvarez-Garc\'ia, Cesar Fierro Cota, Ori Ganor, Arthur Hebecker, Abhiram Kidambi, Daniel Kl\"awer, Alessandro Mininno, Hans Peter Nilles, Enrico Parisini,  Fernando Quevedo and Timo Weigand for useful discussions. J.M.L. and N.R. (partially) are supported by the Deutsche Forschungsgemeinschaft under Germany's Excellence Strategy - EXC 2121 ``Quantum Universe'' - 390833306. N.R. is partially supported by a Leverhulme Trust Research Project Grant RPG-2021-423. A.W. is supported by the ERC Consolidator Grant STRINGFLATION under the HORIZON 2020 grant agreement no. 647995.

\appendix

%=========================================================
\section{Modular Symmetry \& Forms}
\label{app:modularforms}
%=========================================================
Here we collect some basic facts and identities on modular forms relevant for the results and arguments of the main text. For further details, see~\cite{123ModForms}. 

\subsection{Basics \& Definitions}
As stated in the main text, the correct duality group of the K\"{a}hler modulus is  $\pslz=\slz/\{\pm \mathds{1}\}$
since the elements $\gamma\in\slz$ and $-\gamma\in\slz$ define the same transformation for $T$ in~\cref{eq:slztransf}. The generators of $\pslz$ are
\be
\mathcal{S} = \left(\begin{array}{cc}
    0 & 1 \\
    -1 & 0\end{array}\right) \coma
\mathcal{T} = \left(\begin{array}{cc} 1 & 1\\ 0 & 1\end{array}\right)\fstop
\ee
Using these generators, one can show that $T=i$ is stabilized by the order 2 cyclic subgroup $\{1, \mathcal{S}\}$ and 
$T=e^{2\pi i/3}\equiv \rho$ is invariant under the action of the order 3 cyclic subgroup $\{1, \mathcal{S}\mathcal{T},(\mathcal{S}\mathcal{T})^2\}$.

A modular form of weight $(p,q)$ is a function that transforms as
\begin{align}
    f(T,\bar{T})\rightarrow f(\gamma\cdot T,\gamma\cdot \bar{T}) = (cT+d)^p(c\bar{T}+d)^qf(T,\bar{T})\fstop
\end{align}
A particular subclass are holomorphic modular forms with weights $(p,0)$. We give our conventions for these in terms of the nome $q=e^{2\pi i T}$.
\begin{itemize}
    \item j-invariant
          \be\begin{split}
              j(T) =& \frac{1}{q} +744+196884q+21493760q^2 + 864299970q^3\\& + 20245856256q^4+333202640600q^5+\cdots
              \end{split}
        \ee
    \item Modular Discriminant\\
            Is the weight $(12,0)$ cusp form 
            \be
                \Delta (T) = \eta^{24}(T)\coma
            \ee
            where the Dedekind eta is
            \be
                \eta(T) = q^{1/24}\prod_{n=1}^\infty (1-q^n)\fstop
            \ee
            Note that the transformation property of the Dedekind eta is
            \be
                \eta(\gamma\cdot T) = \epsilon(\gamma) (cT+d)^{\frac{1}{2}}\eta(T)\fstop
            \ee
           The multiplier system is a $24$-th root of unity and is given in terms of the Rademacher phi function $\Phi(\gamma)$:
           \be
                \epsilon(\gamma) = \exp\bigg[ \frac{2\pi i}{24}\Phi(\gamma)- \frac{i\pi}{4}\bigg]\fstop
           \ee
           For more explicit formulae, see~\cite{DHoker:2022dxx}.
    \item Holomorphic Eisenstein Series\\
             For $k>1$, the weight $(2k,0)$ holomorphic Eisenstein series are
            \begin{align}
    G_{2k}(T) = \sideset{}{'}\sum_{c,d\in\mathbb{Z}} \frac{1}{(c T +d)^{2k}}
              = \zeta(2k) \sum_{(c,d)=1}\frac{1}{(c T +d)^{2k}} = \zeta(2k) E_{2k}(T)\coma
\end{align}
where the prime indicates that we omit $c,\,d=0$. Our conventions here differ from the mathematics literature by a factor of $2$. These have the Fourier development
        \be
            G_{2k}(T) = 2\zeta(2k)\bigg(1-\frac{4k}{B_{2k}}\sum_{n=1}^\infty\sigma_{2k-1}(n)q^n\bigg)\coma
        \ee
    with the divisor sigma
        \begin{align}
            \sigma_a(n) = \sum_{d | n} d^a \coma
        \end{align}
        where $d|n =$ ``$d$ divides $n$". The exceptional case is $G_2(T)$, which is a quasi-modular form that transforms as   
        \be
            G_2(\gamma\cdot T) = (cT+d)^2G_2(T) -2\pi i c(cT+d)\fstop
        \ee
    Then one can define a non-holomorphic Eisenstein series with weight $(2,0)$ as
        \be
            \widehat{G}_2(T) = G_2(T) + \frac{2\pi}{i(T-\tbar)}\fstop
        \ee
    The quasi-modular form $G_2(T)$ and the Dedekind eta are related by
        \be
            \frac{d}{dT}\ln(\eta(T)) = \frac{i}{4\pi}G_2(T)\fstop
        \ee
    \item Real Analytic Eisenstein Series:\\
            These are non-holomorphic modular forms of weight $(0,0)$. Similar to the series expression of the holomorphic Eisenstein series, we have
            \begin{align}
            \label{eq:RAEis}
                \mathbf{E}_s(T) &= \sum_{\gamma\in\Gamma_{\infty}\backslash PSL(2,\mathbb{Z})} \text{Im}(\gamma\cdot T)^s\\
                                &= t^s + \sum_{(c,d)=1, c\ge 1}\frac{t^s}{\lvert c T+d\rvert^{2s}}\\
                                &= \frac{1}{\zeta(2s)}\mathbf{G}_s(T)\fstop
            \end{align}     
        The sum converges for $\text{Re}(s)>1$ and has a Fourier decomposition
        \begin{align}
            \mathbf{E}_s(T) &= t^s + \frac{\Lambda(1-s)}{\Lambda(s)}t^{1-s}+\sum_{j=1}^\infty 4\cos(2\pi j a) \frac{\sigma_{2s-1}(j)}{j^{s-\frac{1}{2}}\Lambda(s)}\sqrt{t} K_{s-\frac{1}{2}}(2\pi jt)\,,
        \end{align}
    with $K$ is the modified Bessel function of the second kind and $\Lambda(s) = \pi^{-s}\Gamma(s)\zeta(2s)$ is a symmetrized version of the Riemann zeta function. The Fourier decomposition allows for a meromorphic continuation in the whole $s$ plane.
\end{itemize}
In general, the derivative of a modular form is not a modular form. However, one can define a covariant modular derivative utilizing $G_2(T)$. If $M_{2k}(T)$ is a weight $(2k,0)$ modular form, then
\be
M_{2k+2}(T) = \partial_TM_{2k}(T)+\frac{k}{i\pi} G_2(T)M_{2k}(T)
\ee
is a weight $(2k+2,0)$ modular form. $G_2(T)$ is again exceptional, with 
\be
    \widetilde{M}_4(T) = \partial_TG_2(T) -\frac{i}{2\pi}G_2^2(T)
\ee
defining a weight $(4,0)$ modular form. It is also possible to build non-holomorphic modular forms in a similar fashion. If $\widetilde{M}_{\kappa}(T)$ is a weight $(\kappa,0)$ modular form, then \be
    \widetilde{M}_{\kappa+2}(T,\tbar) = \partial_T \widetilde{M}_{\kappa}(T) +\frac{\kappa}{T-\tbar}\,\widetilde{M}_{\kappa}(T)
\label{eq:2ndcov}
\ee
is a weight $(\kappa+2,0)$ non-holomorphic modular form. This is the method by which the scalar potential achieves modular invariance.
There are several features of modular forms that are useful in searching for vacua:
\begin{itemize}
    \item weight $(2k,0)$ modular forms vanish at $T=i$ for $k\in 2\mathbb{Z}+1$,
    \item weight $(2k,0)$ modular forms vanish at $T=\rho$ for $k\neq 0$ mod $3$.
\end{itemize}
These statements follow directly from the transformation properties of the modular forms as the behavior $\mathcal{S}\cdot i =i$ and $\mathcal{S}\mathcal{T}\cdot \rho = (\mathcal{S}\mathcal{T})^2\cdot \rho$. The case $k=1$ applies to the derivative of any modular function and implies that this derivative must vanish at the fixed points, as claimed in the main text for the derivative of~\cref{eq:pot1}.

\subsection{Numerical Values}
At the self-dual points, we have the values
\begin{align}
    \eta(i) &= \frac{\Gamma\left(\frac{1}{4}\right)}{2\pi^{3/4}}\\
                \eta(\rho) &= e^{-\frac{i\pi}{24}}\frac{3^{1/8}\Gamma^{3/2}(1/3)}{2\pi}
\end{align}
and
\begin{align}
    j(i) &= 1728= 12^3\\
    j(\rho) &= 0\\
    \partial_T j(i) &= 0\\
    \partial_T j(\rho) &= 0\fstop
\end{align}
The vanishing of $\partial_T j(T)$ is an avatar of the statements of weight $(2,0)$ forms above. Similarly, we immediately obtain
\be
    \widehat{G}_2(i,-i) = \widehat{G}_2(\rho,\bar{\rho}) = 0\fstop
\ee
For the other Eisenstein series relevant for this paper, we have
\begin{align}
    G_4(i) &= \frac{\Gamma^{8}(1/4)}{960\pi^2}\\
    \partial_T G_4 (i) &= i\frac{\Gamma^{8}(1/4)}{480\pi^2}\\
    G_4(\rho) &= 0\\ 
    \partial_T G_4 (\rho) &=-i\frac{\Gamma^{18}(1/3)}{1280\pi^7}
\end{align}
and
\begin{align}
    G_6(i) &=0\\
    \partial_T G_6(i) &= -i\frac{\Gamma^{16}(1/4)}{215040\pi^5}\\
    G_6(\rho) &= \frac{\Gamma^{18}(1/3)}{8960\pi^6}\\
    \partial_T G_6(\rho) &= i\sqrt{3}\,\frac{\Gamma^{18}(1/3)}{4480\pi^6}\fstop
\end{align}
Thus for the function $H(T)$ in~\cref{eq:Hpara}, we find
\begin{align}
    H(i) &=\begin{cases} 
      0 & m> 0 \\
      \frac{4^n\pi^{4n}}{15^n}\mathcal{P}(1728) & m =  0
   \end{cases}\label{eq:Hcases1}\\
    H_T(i) &=\begin{cases} 
      0 & m> 1 \;\;\&\;\; m=0 \\
      -\frac{i}{7}2^{2n+1}15^{-n-1}\pi^{4n+4}\Gamma^4(1/4)\mathcal{P}(1728) & m= 1
   \end{cases}\label{eq:Hcases2}
\end{align}
and
\begin{align}
    H(\rho) &=\begin{cases} 
      0 & n> 0 \\
      \left(\frac{16i}{35}\right)^m 3^{-\frac{3m}{2}}\pi^{6m}\mathcal{P}(0) & n= 0
   \end{cases}\label{eq:Hcases3}\\
   H_T(\rho) &=\begin{cases} 
      0 & n> 1 \;\;\&\;\; n=0 \\
       -i\left(\frac{16i}{7}\right)^m 3^{-1-\frac{3m}{2}} 5^{-m-1}e^{i\pi/3}\pi^{6m+1}\Gamma^6(1/3)\mathcal{P}(0) & n=1
   \end{cases}\label{eq:Hcases4}
\end{align}
The non-zero values of $H_T(T)$ at the fixed points may seem surprising at first glance given the argument above that the derivative of a modular function must vanish at the fixed points. However, recall that $H(T)$ is not strictly a modular function -- it can transform with a non-trivial multiplier system. Instead, the expression $\partial_T \lvert H\rvert^2 = \bar{H}H_T$ transforms as a proper weight $(2,0)$ non-holomorphic modular form and must vanish at the fixed points. Thus it must be that either $H(T)=0$ or $H_T(T)=0$ at $T=i,\,\rho$.

%=========================================================
\section{K\"{a}hler Modulus Sector Expressions}
\label{app:KahStuff}
%=========================================================
%===========================
\subsection{General Condition For Theorem 1}
\label{app:HTT}
%===========================
In Theorem 1, we referenced an expression for $H_{TT}$ such that $\partial_TV(T,S)=0$ for $\widetilde{F}_T\neq 0$. The expression is   
\begin{align}
    H_{TT} = \,&\frac{3i}{2\pi}\bigg\{ \widehat{G}_2 \widetilde{F}_T+H\partial_T\widehat{G}_2 +\widehat{G}_2H_T\bigg\} + \frac{\mathcal{K}_{TT\bar{T}}}{\mathcal{K}_{T\bar{T}}}\widetilde{F}_T \\&- \left(\frac{3i}{2\pi}\bar{H}\partial_T\bar{\widehat{G}}_2+\mathcal{K}_{T\bar{T}}\bar{H}\left(A(S,\bar{S})-3\right)\right)e^{2i\phi}\coma
\end{align}
where we have suppressed arguments and defined re-scaled F-terms $\widetilde{F}_T = \eta^6 \Omega^{-1}F_T$. The phase is given by $\phi = \text{arg}(\widetilde{F}_T)$.
%===========================
\subsection{General dS Vacuum Conditions for $T=i$}
\label{app:icases}
%===========================
Let us assume that the dilaton subsector is stabilized. By considering the most general case for the polynomial $\mc{P}(j(T))$, we have a dS minimum at $T=i$ when 
\begin{equation}
\begin{array}{l}
     \lvert\mc{B}_n\rvert <1\,:\begin{cases} A(S,\Sb)>2+  \lvert\mc{B}_n\rvert\,\,\,
             \mbox{ if }\,\, \mathcal{R}_n>- \frac{192\pi^4}{\Gamma^8(1/4)}\lvert\mc{B}_n\rvert\\
                \quad\quad\mbox{ or}\\
            \mc{C}_n<A(S,\Sb)<2- \lvert\mc{B}_n\rvert\,\,\,
            \mbox{ if }\,\, \mathcal{R}_n>\frac{192\pi^4}{\Gamma^8(1/4)}\lvert\mc{B}_n\rvert\\ 
            \quad\quad\mbox{ or}\\
            A(S,\Sb)> \mc{C}_n\,\,\, \mbox{ if }\,\, -\frac{192\pi^4}{\Gamma^8(1/4)}<\mathcal{R}_n\leq-\frac{192\pi^4}{\Gamma^8(1/4)}\lvert\mc{B}_n\rvert\\ 
            \quad\quad\mbox{ or}\\
            A(S,\Sb)<\mc{C}_n\,\,\, \mbox{ if }\,\, \mathcal{R}_n<-\frac{192\pi^4}{\Gamma^8(1/4)}
            \end{cases}
        \end{array}
\end{equation}
\begin{equation}
\begin{array}{l}
    \lvert\mc{B}_n\rvert>1\,: \begin{cases}
    2-\lvert\mc{B}_n\rvert<A(S,\Sb)<2+ \lvert\mc{B}_n\rvert\,\,\,
                \mbox{ if }\, \lvert \mathcal{R}_n\rvert<\frac{192\pi^4}{\Gamma^8(1/4)}\lvert\mc{B}_n\rvert\\ 
                \quad\quad\mbox{ or}\\
                \mc{C}_n<A(S,\Sb)<2+ \lvert\mc{B}_n\rvert\,\,\,
                \mbox{ if }\,\, \mathcal{R}_n\geq\frac{192\pi^4}{\Gamma^8(1/4)}\lvert\mc{B}_n\rvert\\ 
                % \quad\quad\mbox{ or}\\
                % 2-\mc{B}<A(S,\Sb)<\mc{C} \,\,\,\mbox{ if } 1+8n+41472\,\re\!\left(\!{\frac{\mc{P}^\prime(1728)}{\mc{P}(1728)}}\!\right)\geq \frac{-192\pi^4}{\Gamma^8(1/4)}\\ 
                \quad\quad\mbox{ or}\\
                2- \lvert\mc{B}_n\rvert <A(S,\Sb)<\mc{C}_n\,\,\, \mbox{ if }\,\, \mathcal{R}_n<-\frac{192\pi^4}{\Gamma^8(1/4)}\lvert\mc{B}_n\rvert
                \end{cases}
\end{array}
\end{equation}
where we used the notation for $\mc{B}_n$ of \cref{eq:definitionB} and we have also introduced
\be
\mathcal{R}_n = 1+8n+41472\, \re\!\left({\frac{\mc{P}^\prime(1728)}{\mc{P}(1728)}}\right)\coma
\ee
\be
\mc{C}_n\equiv\frac{2+\frac{\Gamma^8(1/4)}{192\pi^4}\,\mathcal{R}_n-\lvert\mc{B}_n\rvert^2}{1+\frac{\Gamma^8(1/4)}{192\pi^4}\,\mathcal{R}_n}\fstop
\ee
If we assume $\mc{P}(1728)\in \mathbb{R}$, these conditions greatly simplify to
 \be
    \begin{array}{lcl}
         2-\mc{B}_n<A(S,\bar{S})<2+\mc{B}_n & \quad\mbox{ for }\quad & \mc{B}_n>1\coma \\
        A(S,\bar{S})>2+\lvert\mc{B}_n\rvert & \quad\mbox{ for }\quad & -1<\mc{B}_n<1\coma \\
         2+\mc{B}_n<A(S,\bar{S})<2-\mc{B}_n & \quad\mbox{ for }\quad & \mc{B}_n<-1 \fstop
    \end{array}
\ee

%=========================================================
\section{Orbifolds with $\delta_{GS}\neq 0$ \& A Third No-Go Theorem}
\label{sec:anomaly}
%=========================================================
A natural question arises from the discussion in the main text -- can the mixing of the moduli in the K\"{a}hler potential evade the theorems and result in dS vacua, even in the absence of Shenker-like terms? There is a natural basis for this mixing arising from anomaly cancellation, as described in~\cref{sec:heterotic}. We now return to these orbifold models and prove a limited no-go result similar to Theorem 2.

We will utilize the chiral multiplet formalism where the dilaton is invariant under modular transformations. Our starting points are the superpotential~\cref{eq:superpot1} and the K\"{a}hler potential~\cref{eq:kahpot1} using~\cref{eq:modinvKahpot} so that  
\be
   k(S,\Sb,T,\tbar) = -\ln\left(S+\bar{S}+\delta_{GS}\ln(-i(T-\tbar)\lvert\eta(T)\rvert^4) \right)+ \delta k(S,\Sb,T,\tbar) \fstop
\label{eq:biganomKah}
\ee
Note that we have defined the dilaton-sector K\"{a}hler potential with the loop-corrected contribution and a function $\delta k(S,\Sb,T,\tbar)$, which we take to be a modular-invariant, non-holomorphic function that encodes non-perturbative effects such as the Shenker-like terms. The Green-Schwarz anomaly coefficient $\delta_{GS}$ should not be confused with $\delta k$. 
Similar to the main text, the F-terms read
\begin{align}
    F_S &= \frac{H(T)}{\eta^6(T)}\bigg( \Omega_S + (-Y^{-1}+ \delta k_S) \Omega\bigg)\equiv \frac{H(T)}{\eta^6(T)}\widetilde{F}_S\coma \label{eq:GenDilF}\\
    F_T &= \frac{\Omega(S)}{\eta^6(T)}\bigg( H_T - \frac{3i}{2\pi}\bigg(1+\frac{\delta_{GS}}{3Y}\bigg)H\widehat{G}_2+\delta k_T H\bigg) \equiv \frac{\Omega(S)}{\eta^6(T)}\widetilde{F}_T\coma
    \label{eq:GenKahF}
\end{align}
where we have defined $Y = S+\Sb+\delta_{GS}\ln(-i(T-\tbar)\lvert\eta(T)\rvert^4)$.
The K\"{a}hler metric is much more complicated than the cases found in the main text and we will not reproduce here. We write the scalar potential for these models as
\be
\begin{split}
V &= e^kZ(T,\tbar)\bigg\{ \lvert H\rvert^2 
    \lvert\Omega\rvert^2 (\widetilde{A}(S,T)-3)+\lvert\Omega\rvert^2\widehat{V}_2(S,T)\\ 
    &\hspace{3.7cm} +H\bar{\Omega}\mathcal{K}^{S\tbar}\widetilde{F}_S\bar{\widetilde{F}}_{\bar{T}} +\bar{H}\Omega\mathcal{K}^{\Sb T}\bar{\widetilde{F}}_{\bar{S}}\widetilde{F}_T\bigg\}\coma
    \end{split}
\label{eq:genpot}
\ee
with $\widetilde{A}(S,T) =  \lvert\Omega\rvert^{-2}\mathcal{K}^{S\Sb}\widetilde{F}_S\bar{\widetilde{F}}_{\bar{S}}$, $\widehat{V}_2(S,T) = \mathcal{K}^{T\tbar}\widetilde{F}_T\bar{\widetilde{F}}_{\bar{T}}$, and $Z(T,\tbar)$ as defined in~\cref{eq:potparts3}.

We now prove a no-go result similar to Theorem 2 of the main text. In particular, we demonstrate the impossibility of dS vacua at the fixed points of $\pslz$ if the $\delta k$ contribution to the K\"{a}hler potential is discarded.

\begin{theorem}
At the points $(T_0=i,S_0)$ or $(T_0=\rho,S_0)$, the scalar potential in~\cref{eq:genpot} defined via~\cref{eq:superpot1} and~\cref{eq:biganomKah} with $\delta k=0$
can not simultaneously satisfy:
\begin{enumerate}[label=(\roman*).,leftmargin=3\parindent]
    \item $V(T_0,S_0) > 0$
    \item $\partial_SV(T_0,S_0)= 0 \quad \&  \quad \partial_TV(T_0,S_0) = 0 $
   % \item $(\Omega_S+k_S\Omega)\rvert_{S=S_0}\neq0$
    %\item $\widetilde{F}_T(T_0)=0$
    \item Eigenvalues of the Hessian of $V(T,S)$ at $(T_0,S_0)$ are all $\ge 0$.\vspace{0.2cm}
\end{enumerate}
\end{theorem}
%\textit{Proof:}
\begin{proof}
First, we are considering the fixed points of $\pslz$, so the requirement $\partial_TV(T_0,S_0) = 0$ in (ii) is automatically satisfied. Furthermore, the Hessian is also block diagonal, so we can again consider the K\"{a}hler modulus and dilaton blocks individually.    
Second, recall from~\cref{app:modularforms} that at the fixed points of $\pslz$, the product $(\bar{H}H_T)_0$ must vanish as a consequence of modular symmetry -- $\partial_T \lvert H\rvert^2$ transforms as a weight $(2,0)$ modular form with trivial multiplier system. One can also observe this numerically from~\cref{eq:Hcases1,eq:Hcases2,eq:Hcases3,eq:Hcases4}. In almost all cases, the vanishing of this product occurs because $H_T=0$. This immediately implies that $\widetilde{F}_T=0$ since $\widehat{G}_2(T,\tbar)$ must also vanish at the fixed points. However, cases where the integer pair $(m,n)$ in~\cref{eq:Hpara} correspond to $(1,n)$ or $(m,1)$ do not follow this pattern since they give $H_T\neq 0$ at $T=i,\,\rho$. Thus, we must divide our analysis into two cases. We will not consider the cases where $H=H_T=0$ at $T=i,\,\rho$ since these trivially give Minkowski extrema.\\

\noindent\textbf{Case 1:} $H(T_0)\neq 0$\\
\noindent This is the case where $(\widetilde{F}_T)_0 =0$, so all positive contributions to the vacuum energy must come from the dilaton sector in the form of the function $\widetilde{A}(S,T)$ in~\cref{eq:genpot}. Thus to satisfy (i), we require $(\widetilde{F}_S)_0\neq 0$ \& $(\mathcal{K}^{S\Sb})_0\neq 0$. We can again introduce a parameter $\Lambda$ such that 
\be
    e^{k_0}Z_0\lvert H_0\rvert^2 \lvert\Omega_0\rvert^2(\widetilde{A}_0-3) = e^k_0Z_0\lvert H_0\rvert^2\Lambda^4\fstop
\ee
This can be solved to give
\be
    \Omega_S(S_0) = Y_0^{-1}\Omega_0 \pm (\mathcal{K}^{S\Sb}_0)^{-1/2}\left(\Lambda^2 \pm i\sqrt{3\lvert H_0\rvert^2}\right)\fstop
\ee
Once again any of the four sign choices are valid. The first condition in (ii) is satisfied if
\be
    \Omega_{SS}(S_0) = \frac{2\bar{\Omega}_0}{Y_0^2}\frac{(\bar{\widetilde{F}}_{\bar{S}})_0}{(\widetilde{F}_S)_0}\coma
\ee
where we have utilized the vanishing of $(\widetilde{F}_T)_0$. Inserting the above expressions into $\partial_S\partial_{\bar{S}}V(S_0,T_0)$, we find
\be
    \partial_S\partial_{\bar{S}}V(S_0,T_0) = - \frac{2 Z_0 \lvert H_0\rvert^2 (\Lambda^4+2\lvert\Omega_0\rvert^2)}{Y_0^3}\fstop
\ee
Note that $Y_0>0$ is required for a sensible value of the string coupling. Hence, by identical logic to Theorem 2, at least one eigenvalue in the dilaton block of the Hessian is negative and a (meta)stable dS vacuum is not possible for any function $\Omega(S)$.\\

\noindent\textbf{Case 2:} $H_T(T_0)\neq 0$\\
Since it must be that $H_0=0$ for this case, the potential in~\cref{eq:genpot} simplifies greatly to
\be
    V_0 = e^{k_0}Z_0 \lvert\Omega_0\rvert^2 (\widehat{V}_2)_0\fstop
\ee
Thus we must demand that $\Omega_0\neq 0$ to avoid a Minkowski extremum. Next, we note that
\be
    \partial_SV(S_0,T_0) = \frac{-(T_0-\tbar_0)^2}{3Y_0+\delta_{GS}}Z_0\bar{\Omega}_0 \lvert H_T\rvert^2_0\left((3Y_0+\delta_{GS})(\Omega_S)_0 -3\Omega_0\right)\coma
\ee
after using the vanishing of $H_0$. Then to satisfy the first condition of (ii), the factor in parentheses must vanish and we obtain $(\Omega_S)_0 = 3(3Y_0+\delta_{GS})^{-1}\Omega_0$. We now investigate the K\"{a}hler modulus block of the Hessian. Using the above,\footnote{We also use $(H_{TT})_0 = -2(T_0-\tbar_0)^{-1}(H_T)_0$. This follows from~\cref{eq:2ndcov} with $\widetilde{M}_2 = H_T$ and so $\bar{H}_{\bar{T}}(H_{TT}+2(T-\tbar)^{-1}H_T)$ transforms as a weight $(4,2)$ modular form with trivial multiplier system. Thus this combination must vanish at the fixed points, and since we are considering cases here where $(H_T)_0\neq 0$, the identity follows.} we see that
\be
\partial_T\partial_{\bar{T}}V(S_0,T_0) = -\frac{2(3Y_0^2+\delta_{GS}Y_0+\delta_{GS}^2)}{Y_0(3Y_0+\delta_{GS})^2}Z_0\lvert\Omega_0\rvert^2\lvert H_T\rvert^2_0\fstop
\ee
Since both $Y>0$ and $\delta_{GS}>0$~\cite{Derendinger:1991hq}, we see that  $\partial_T\partial_{\bar{T}}V(S_0,T_0)<0$ and by the same logic used in Theorem 1 at least one of the eigenvalues of the K\"{a}hler modulus block of the Hessian is negative.
\end{proof}

\noindent From Theorem 3, we conclude that there is no non-perturbative contribution in the form of $\Omega(S)$ or $H(T)$ that can lift a fixed point extremum to a dS vacuum if $\delta k=0$, even for models with $\delta_{GS}\neq 0$. As discussed in the main text, if this additional term is included, dS vacua may be possible. If we again consider a fixed point $T_0$ of $\pslz$ with $H(T_0)\neq 0$, then~\cref{eq:genpot} is simply
\be
    V(S,T_0) = e^k Z(T_0,\tbar_0)\lvert\Omega\rvert^2\lvert H(T_0)\rvert^2 (\widetilde{A}(S,T_0)-3)
\ee
This is similar to the scenarios considered in the main text. We will not perform a thorough analysis here, but one could determine stability criteria for dS vacua at the fixed points and convert them into requirements on the Shenker-like terms encoded in $\delta k(S,\Sb,T,\tbar)$.

%=========================================================
\section{Linear Multiplet Formalism \& The Inverse Map}
%=========================================================
In this appendix, we develop a perturbative framework to invert the map between the linear multiplet formalism and the chiral multiplet formalism for the dilaton.
\label{app:linearchiral}
With the Shenker-like terms corrections, we have the relation
\begin{align}
 \frac{\ell}{1+f(\ell)}= \frac{1}{S+\bar{S}}\fstop
 \label{eq:invert1}
\end{align}
Our goal is to invert this relationship and write $\ell = \ell(s)$. Recalling the form of $f(\ell)$ used in~\cref{eq:npcorrlinear}
\begin{align}
    f(\ell) = \sum_n A_n \ell^{q_n}e^{-B/\sqrt{\ell}}\coma
\end{align}
we see why the task is non-trivial -- inversion requires us to solve a transcendental equation. We can make some progress in developing a perturbative approach. Let us work in the limit $\ell\ll 1$. Then the exponential dies out and to the lowest order we simply have $2s=\ell^{-1}$. To proceed, we want to develop a series 
\be
    \ell = \frac{1}{2s}+ \epsilon y(s) + \cdots
\ee
where $\epsilon$ is a small parameter controlling the expansion and $y(s)$ is some function of the dilaton to be determined. Unfortunately a brute force perturbative approach does not work -- we cannot expand the exponential $e^{-B/\sqrt{\ell}}$ about $\ell = 0$ since it has an essential singularity there. In the limit we are considering, $e^{-B/\sqrt{\ell}}$ is small, so our perturbative expansion should be based on the notion that
\be
    e^{-B/\sqrt{\ell}} \sim \mathcal{O}(\epsilon)\fstop
\ee
This would then imply that 
\be
    \ell \sim \frac{B^2}{\ln^2(\epsilon)}
\ee
to leading order. This is quite a peculiar setup, but intuitively it seems reasonable -- if $\epsilon$ is a small parameter, then $\ln^{-2}(\epsilon)$ is also a small positive number. We are still performing a perturbative expansion, just one that involves a non-linear map of the expansion parameter.
We formalize this observation by postulating a very general perturbative expansion 
\begin{align}
\begin{split}
    \ell &= \alpha \ell_1 + \epsilon^{a_2}\alpha^{b_2}\ell_2 + \epsilon^{a_3}\alpha^{b_2}\ell_3+\cdots\\
        &= \alpha \ell_1 + \delta\coma
\end{split}
\end{align}
where $\alpha = \ln^{-2}(\epsilon^a)$ and $a$ is unspecified for the moment. We also leave the powers in the higher order terms as variables to be solved. %Note that this perturbative expansion does make intuitive sense -- if $\epsilon<<1$, then $\ln^{-2}(\epsilon)$ is also small, but bigger than $\epsilon$. 
We now plug this into~\cref{eq:invert1}, expand around $\delta=0$, and match powers of $\epsilon$ and $\alpha$. For a sensible expansion, we must relate $\epsilon$ to $s$. The natural guess from the discussion above is that
\begin{align}
    s = \ln^2(\epsilon^a)\fstop
\end{align}
%Thus if $\epsilon$ is small, then $s=\ln^2(\epsilon)$ will be large, as required. 
This also implies that the small parameters $\epsilon$ and $\alpha$ are defined in terms of $s$ as
\begin{equation}
    \epsilon = e^{-\frac{\sqrt{s}}{a}}\coma \mbox{where }\,\,
    \alpha = \frac{1}{s}\fstop
\end{equation}
To determine the value of $a$, we demand that
\begin{align}
    e^{-B/\sqrt{\ell_1}} &= \epsilon\coma\\
    \epsilon^{\sqrt{\ell_1}Ba} &= \epsilon\coma
\end{align}
so $a = \frac{\sqrt{\ell_1}}{B}$.
\subsection{Example 1: Only $A_0\neq 0$}
We explicitly solve the simple example where the series in~\cref{eq:npcorrlinear} collapses to $f(\ell) = A_0 e^{-B/\sqrt{\ell}}$. Then we have
\begin{align}
    2s &= \ell^{-1}(1+A_0e^{-B/\sqrt{\ell}})\\
    2\ln^2(\epsilon^a) &= (\alpha\ell_1+\delta)^{-1}(1+A_1e^{-B/\sqrt{\alpha\ell_1+\delta}})\fstop
\end{align}
To leading order,
\be
    \begin{split}
    2\ln^2(\epsilon^a) &= \frac{\ln^2(\epsilon^a)}{\ell_1}\\
    \ell_1 &= \frac{1}{2}\fstop
    \end{split}
\ee
Matching the next order yields
\be
    \begin{split}
    (a_2,b_2) &= (1,1)\\
    \ell_2 &= A_0\ell_1= \frac{A_0}{2}
    \end{split}
\ee
and following order is
\be
    \begin{split}
        (a_3,b_3) &= (2,1/2)\\
        \ell_3 &= \frac{A_0B\ell_2}{2\sqrt{\ell_1}}= \frac{A_0^2B}{2\sqrt{2}}\fstop
    \end{split}
\ee
We then arrive at the series expansion
\be
    \begin{split}
    \ell &= \alpha \ell_1 + \epsilon\alpha\ell_2 + \epsilon^2\alpha^{1/2}\ell_3 +\cdots  \\
         &=  \frac{\ell_1}{s} + \frac{e^{-B\sqrt{s/\ell_1}}}{s}\ell_2 + \frac{e^{-2B\sqrt{s/\ell_1}}}{s^{1/2}}\ell_3 +\cdots\\
         &= \frac{1}{2s} + \frac{A_0}{2s}e^{-B\sqrt{2s}} + \frac{A_0^2B}{2\sqrt{2s}}e^{-2B\sqrt{2s}}+\cdots\fstop
    \end{split}     
\ee

\subsection{Example 2: $A_0$, $A_1\neq 0$}
We now try the slightly harder example where $A_0,A_1\neq 0$ and $q_1=-\frac{1}{2}$ in~\cref{eq:npcorrlinear}. We must then solve
\be
    2s = \ell^{-1}(1+ (A_0+A_1\ell^{-1/2})e^{-B/\sqrt{\ell}})\fstop
\ee
Much of our discussion from the previous example carries over -- the only adjustment we have to make is to re-do the matching to determine values of the $(a_n,b_n)$. We find
\be
    \begin{split}
    (a_2,b_2) &= (1,1)\\
    \ell_2 &= A_0\ell_1
    \end{split}
\ee
\be
    \begin{split}
    (a_3,b_3) &= (1,1/2)\\
    \ell_3 &= A_1\ell_1^{1/2}
    \end{split}
\ee
Thus including more terms in the polynomial defining $f(\ell)$ results in slightly different perturbative expansions -- this seems a direct result of the new powers introduced by more terms in the polynomial. Our perturbtaive expansion is now
\be
    \begin{split}
    \ell &= \alpha\ell_1 + \epsilon\alpha\ell_2 + \epsilon\alpha^{1/2}\ell_3+\cdots\\
        &= \frac{1}{2s} +\frac{A_0}{2s}e^{-B\sqrt{2s}} + \frac{A_1}{\sqrt{2s}}e^{-B\sqrt{2s}}+\cdots
    \end{split}
\ee

\bibliographystyle{JHEP}
\bibliography{ref}

\end{document}